%% file: arxiv.tex
\documentclass[11pt]{article}
\usepackage[margin=1in]{geometry}
\usepackage[utf8]{inputenc}

\usepackage{booktabs,amsthm,amsmath,amssymb} 	
\usepackage{longtable}
\usepackage{bm}
\usepackage{hyperref} 
\usepackage{graphicx} 	
\usepackage{float,wrapfig} 		    
\usepackage{pgfplots} 	    
\pgfplotsset{compat=1.7}    
\usepackage{mathtools,verbatim}      
\usepackage{caption}
\usepackage{subcaption}
\usepackage{color}
\usepackage{colortbl}
\usepackage{multicol}
\usepackage{xcolor}  
\definecolor{lightgreen}{rgb}{0.56, 0.93, 0.56}  
\definecolor{lightblue}{rgb}{0.68, 0.85, 0.90}
\usepackage{enumerate}
\usepackage{enumitem}
\usepackage[appendix=append]{apxproof}
\newtheoremrep{theorem}{Theorem}
\newtheoremrep{lemma}{Lemma}

\newlist{Ren}{enumerate}{3}          

\setlist[Ren,1]{label=\textbf{(R\arabic{Reni})},
                ref=R\arabic{Reni},
                leftmargin=*, itemsep=4pt}

\setlist[Ren,2]{label=\textbf{(R\arabic{Reni}\alph{Renii})},
                ref=R\arabic{Reni}\alph{Renii},
                leftmargin=*, labelindent=-1em, itemsep=4pt}

\setlist[Ren,3]{label=\textbf{(R\arabic{Reni}\alph{Renii}-\roman{Reniii})},
                ref=R\arabic{Reni}\alph{Renii}-\roman{Reniii},
                leftmargin=*, labelindent=-2em, itemsep=4pt}

\usepackage[ruled, vlined,linesnumbered]{algorithm2e}
\usepackage{tikz}
\usetikzlibrary{positioning, shapes.geometric, arrows.meta, calc, matrix,fit, shapes, shapes.misc, shapes.symbols}
\def\BibTeX{{\rm B\kern-.05em{\sc i\kern-.025em b}\kern-.08em
    T\kern-.1667em\lower.7ex\hbox{E}\kern-.125emX}}

\newboolean{short}
\setboolean{short}{true}
\newcommand{\shortOnly}[1]{\ifthenelse{\boolean{short}}{#1}{}}
\newcommand{\onlyShort}[1]{\ifthenelse{\boolean{short}}{#1}{}}
\newcommand{\longOnly}[1]{\ifthenelse{\boolean{short}}{}{#1}}
\newcommand{\onlyLong}[1]{\ifthenelse{\boolean{short}}{}{#1}}

\newif\ifshowtikz
\showtikztrue

\newcommand{\sync}{$\mathcal {SYNC}$}
\newcommand{\async}{$\mathcal {ASYNC}$}
\newcommand{\dis}{\textsc{dispersion}}

\newtheorem{definition}{Definition}
\newtheorem{observation}{Observation}
\newtheorem{remark}{Remark}

\begin{document}
\title{\bf Optimal Dispersion Under Asynchrony}

\author{
Debasish Pattanayak$^1$, Ajay D. Kshemkalyani$^2$, Manish Kumar$^3$, Anisur Rahaman Molla$^4$, \\
Gokarna Sharma$^5$ \\\\
   $^1$Indian Institute of Technology Indore, India\\
   $^2$University of Illinois Chicago, USA\\
   $^3$Indian Institute of Technology Madras, India\\
   $^4$Indian Statistical Institute, Kolkata, India\\
 $^5$Kent State University, USA
}



\date{}

\maketitle
\begin{abstract}
We study the dispersion problem in anonymous port-labeled graphs: $k \leq n$ mobile agents, each with a unique ID and initially located arbitrarily on the nodes of an $n$-node graph with maximum degree $\Delta$, must autonomously relocate so that no node hosts more than one agent. Dispersion serves as a fundamental task in distributed computing of mobile agents, and its complexity stems from key challenges in local coordination under anonymity and limited memory. 

The goal is to minimize both the time to achieve dispersion and the memory required per agent. It is known that any algorithm requires $\Omega(k)$ time in the worst case, and $\Omega(\log k)$ bits of memory per agent. A recent result [SPAA'25] gives an optimal $O(k)$-time algorithm in the synchronous setting and an $O(k \log k)$-time algorithm in the asynchronous setting, both using $O(\log(k+\Delta))$ bits.

In this paper, we close the complexity gap in the asynchronous setting by presenting the first dispersion algorithm that runs in optimal $O(k)$ time using $O(\log(k+\Delta))$ bits of memory per agent. Our solution is based on a novel technique we develop in this paper that constructs a port-one tree in anonymous graphs, which may be of independent interest.

\end{abstract}

\section{Introduction}
\label{section:intro}
The problem of dispersion of mobile agents studied extensively in recent distributed computing literature not only takes its inspiration from biological phenomena, such as damselfish establishing non-overlapping territories on coral reefs \cite{brawley1977territorial}, or neural crest cells migrating and distributing themselves across the developing embryo \cite{le1999neural}; but also with practical applications such as placing a fleet of small autonomous robots (agents) under shelves (nodes) in fulfillment centers~\cite{wurman2008coordinating}.
It is also closely connected to other coordination tasks such as
exploration, scattering, load balancing, and self-deployment~\cite{Barriere2009,Cybenko:1989,Das18,Dereniowski:2015,Fraigniaud:2005,Subramanian:1994}.  
The {\em dispersion} problem, denoted as {\dis}, involves $k\leq n$ mobile
 agents placed initially arbitrarily on the nodes of an $n$-node anonymous graph of maximum degree $\Delta$. The goal for the agents is to autonomously relocate such that each agent is on a distinct node of the graph (see Fig.~\ref{fig:dispersion}).
 The objective is to design algorithms that simultaneously optimize time and memory complexities. 
 Time complexity is the total time required to achieve dispersion starting from any initial configuration. Memory complexity is the maximum number of bits stored in the persistent memory at each agent throughout the execution. We stress that graph nodes are memory-less and cannot store any information. Fundamental performance limits exist: certain graph topologies (e.g., line graphs) necessitate $\Omega(k)$ time for dispersion. Concurrently, $\Omega(\log k)$ memory bits per agent is required, at minimum for storing unique identifiers.

{\dis} has been studied in a series of papers, e.g.,  \cite{Augustine:2018,DasCALDAM21,GorainSSS22,ItalianoPS22,KshemkalyaniICDCN19,kshemkalyani2025dispersion,KshemkalyaniALGOSENSORS19,KshemkalyaniWALCOM20,KshemkalyaniICDCS20,KshemkalyaniICDCN20,KshemkalyaniOPODIS21,tamc19,MollaIPDPS21,MollaMM21} (see Table  \ref{table:comparision}).
The state-of-the-art is the two recent results due to Kshemkalyani {\it et al.} \cite{kshemkalyani2025dispersion}, one in the synchronous setting and another in the asynchronous setting. In the {\em synchronous} setting ({\sync}), 
all agents perform their operations simultaneously in synchronized rounds (or steps), and hence time complexity (of the algorithm) is measured in rounds. However, in the {\em asynchronous} setting ({\async}), agents become active at arbitrary times and perform their operations in arbitrary duration, and hence time complexity is measured in epochs -- an {\em epoch} represents the time interval in which each agent becomes finishes at least one cycle of execution. In {\sync}, an epoch is a round.  
In particular, the {\sync} algorithm of Kshemkalyani {\it et al.} \cite{kshemkalyani2025dispersion} has time complexity optimal $O(k)$ rounds and memory complexity $O(\log (k+\Delta))$ bits per agent.  Their {\async} algorithm \cite{kshemkalyani2025dispersion} has time complexity $O(k\log k)$ epochs and memory complexity $O(\log (k+\Delta))$ bits per agent. Time and memory complexities of Kshemkalyani {\it et al.} \cite{kshemkalyani2025dispersion} apply to both rooted (all $k$ agents initially at a node) and general initial configurations ($k$ agents initially on multiple nodes). 

\input{fig_initial}

\begin{table*}[!t]
\centering
\resizebox{\textwidth}{!}{  
\begin{tabular}{cccc}
\toprule
\textbf{Algorithm} & \textbf{Memory/agent (in bits)} & \textbf{Time (in rounds/epochs)} & \textbf{Rooted/General} \\
\toprule
Lower bound & $\Omega(\log k)$ \cite{KshemkalyaniICDCN19} & $\Omega(k)$ \cite{KshemkalyaniICDCN19} & any \\
\midrule
\multicolumn{4}{c}{\textbf{Synchronous Algorithms}} \\
\midrule
\cite{KshemkalyaniALGOSENSORS19}$^\dagger$ & \cellcolor{lightgreen}$O(\log(k+\Delta))$ & $O(\min\{m,k\Delta\} \cdot \log k)$ & any \\
\cite{ShintakuSKM20} & \cellcolor{lightgreen}$O(\log(k+\Delta))$ & $O(\min\{m,k\Delta\} \cdot \log k)$ & any \\
\cite{KshemkalyaniICDCN20} & $O(D + \Delta \log k)$ & $O(D\Delta(D + \Delta))$ & rooted \\
\cite{sudo24} & \cellcolor{lightgreen}$O(\log \Delta)$ & $O(k \cdot \log k)$ & rooted \\
\cite{sudo24} & $O(\Delta)$ & \cellcolor{lightblue}$O(k)$ & rooted \\
\cite{sudo24} & \cellcolor{lightgreen}$O(\log(k+\Delta))$ & $O(k \cdot \log^2 k)$ & general \\
\cite{kshemkalyani2025dispersion} & \cellcolor{lightgreen}$O(\log(k+\Delta))$ & \cellcolor{lightblue}$O(k)$ & any \\
\midrule
\multicolumn{4}{c}{\textbf{Asynchronous Algorithms}} \\
\midrule
\cite{Augustine:2018} & $O(k \log(k+\Delta))$ & $O(\min\{m,k\Delta\})$ & any \\
\cite{KshemkalyaniICDCN19} & \cellcolor{lightgreen}$O(\log(k+\Delta))$ & $O(\min\{m,k\Delta\} \cdot k)$ & any \\
\cite{KshemkalyaniOPODIS21} & \cellcolor{lightgreen}$O(\log(k+\Delta))$ & $O(\min\{m,k\Delta\})$ & any \\
\cite{kshemkalyani2025dispersion} & \cellcolor{lightgreen}$O(\log(k+\Delta))$ & $O(k \cdot \log k)$ & any \\
\textbf{This paper} & \cellcolor{lightgreen}$\bm{O(\log(k+\Delta))}$ & \cellcolor{lightblue}$\bm{O(k)}$ & \textbf{any} \\
\bottomrule
\end{tabular}
}
\caption{{\dis} of $k \leq n$ agents on an anonymous $n$-node $m$-edge graph of diameter $D$ and maximum degree $\Delta$. $^\dagger$Requires knowledge of $m, k, n, \Delta$. Optimal memory cells are highlighted in green, and optimal time cells are highlighted in blue.}
\label{table:comparision}
\vspace{-4mm}
\end{table*}

\vspace{2mm}
\noindent{\bf Contributions.} In light of the state-of-the-art results from  Kshemkalyani {\it et al.} \cite{kshemkalyani2025dispersion}, the following question naturally arises: {\em Can optimal $O(k)$-epoch solution be designed for {\dis} in {\async}?}  Such  a contribution would complete the picture of optimal solutions to {\dis}.  
In this paper, we  answer this question in the affirmative by providing an optimal $O(k)$-epoch solution in {\async} with $O(\log(k+\Delta))$ bits per agent.  
Our result shows that synchrony assumption is irrelevant for time optimal dispersion. 

The result is possible from a novel construction of a special tree, which we call a {\em Port-One tree} (denoted as {\sc P1Tree}), that we introduce in this paper. {\sc P1Tree} prioritizes the edges with port-1 at either of its endpoints. 
This prioritization helps to find the  \emph{empty} neighbor nodes of a node (if exist) to settle agents in $O(1)$ epochs even in {\async} with $O(\log(k+\Delta))$ memory. Kshemkalyani {\it et al.} \cite{kshemkalyani2025dispersion} were able to do so only in {\sync} and their approach does not extend to {\async}.  Since $k$ agents can settle solving dispersion after visiting $k$ empty nodes and visiting each such empty nodes takes $O(1)$ epochs, the whole algorithm finishes in optimal $O(k)$ epochs.  The {\sc P1Tree} concept may be useful in solving many other coordination problems in anonymous port-labeled graphs optimizing time and memory complexities. 

\vspace{2mm}
\noindent{\bf Challenges.} 
Existing {\dis} algorithms largely relied on breadth-first-search (BFS) and depth-first-search (DFS) techniques, with preference for DFS due to its advantage for optimizing memory complexity along with time complexity.
Given $k$ agents in a rooted initial configuration, DFS starts in the forward phase and works alternating between forward and backtrack phases until $k-1$ empty nodes are visited to solve {\dis}. 
During each forward phase, one such empty node becomes settled. To visit $k$ different empty nodes, a DFS must perform at least $k-1$ forward phases, and at most $k-1$ backtrack phases. Therefore, the best DFS time complexity is $2(k-1)=O(k)$, which is asymptotically optimal since in graphs (such as a line graph) exactly $k-1$ forward phases are needed in the worst-case (consider the case of all $k$ agents are on the either end node of the line graph).
Therefore, the challenge is to limit the traversal to exactly $k-1$ forward phases and finish each in $O(1)$ time to obtain $O(k)$ time bound.
Suppose DFS is currently at a node. 
To finish a forward phase at a node in $O(1)$ time, 
an empty neighbor node (if exists) of that node need to be found in $O(1)$ time; an empty node is the one that has no agent positioned on it yet.  
The state-of-the-art result of Kshemkalyani {\it et al.} \cite{kshemkalyani2025dispersion} developed a technique to find an empty neighbor of a node in $O(1)$ rounds in {\sync} and $O(\log \min\{k,\Delta\}{})$ epochs in {\async}. 
This technique allowed Kshemkalyani {\it et al.} \cite{kshemkalyani2025dispersion} 
 to achieve $O(k)$-round solution in {\sync} and $O(k\log k)$-epoch solution in {\async} in rooted initial configurations. 

 In general initial configurations, let $\ell$ be the number of nodes with two or more agents on them, which we call {\em multiplicity nodes} 
(for the rooted case, $\ell=1$). There will be $\ell$ DFSs initiated from those $\ell$ multiplicity nodes. It may be the case that two or more DFSs {\em meet}. The meeting situation needs to be handled in a way that ensures it does not increase the time required to find empty neighbor nodes. Kshemkalyani {\it et al.} \cite{kshemkalyani2025dispersion}  handled such meetings in additional time proportional to $O(k)$ in {\sync} and $O(k\log k)$ in {\async} using the size-based subsumption technique of Kshemkalyani and Sharma \cite{KshemkalyaniOPODIS21}. This allowed  Kshemkalyani {\it et al.} \cite{kshemkalyani2025dispersion} 
 to achieve an $O(k)$-round solution in {\sync} and an $O(k\log k)$-epoch solution in {\async} even starting from general initial configurations.

To have an $O(k)$-epoch solution in {\async} (for rooted and general initial configurations), we need a technique to find an empty neighbor node (if exists) of a node in $O(1)$ epochs. A natural direction is to explore whether Kshemkalyani {\it et al.} \cite{kshemkalyani2025dispersion}'s technique can be extended to {\async}. The major obstacle on doing so is the technique of {\em oscillation} used in Kshemkalyani {\it et al.} \cite{kshemkalyani2025dispersion}.
To have sufficient agents to explore neighbors in $O(1)$ time, Kshemkalyani {\it et al.} \cite{kshemkalyani2025dispersion} leave $\lceil k/3\rceil$ nodes visited by DFS empty, which we call {\em vacated} nodes. The $\lceil k/3\rceil$ agents that were supposed to settle at those vacated nodes are then used in neighborhood search. This approach created one problem: while probing at a node, how to differentiate an empty neighbor node from a vacated neighbor.  Kshemkalyani {\it et al.} \cite{kshemkalyani2025dispersion} overcame this difficulty as follows. They selected vacated nodes in such a way that there is an occupied node (with an agent on it) that makes an oscillation trip of length $6$ in time $6$ rounds to cover the vacated nodes assigned to it. The probing agent waits at the (possibly empty) neighbor for 6 rounds before returning. 
While waiting at a vacated node, it is guaranteed that an oscillating agent reaches that vacated node during those 6 rounds. For empty node, no such agent reaches that node.      
It is easy to see that why this approach does not extend to  {\async}: Making decision by probing agents on how long to wait at an empty neighbor for an oscillating agent to arrive (or not arrive) is difficult since agents do not have agreement on duration and the start/end of each computational step (i.e., no agreed-upon round definition) due to asynchrony.

The novel construction of a Port-One Tree ({\sc P1Tree}) in this paper allowed us to obviate the need of oscillation to differentiate empty neighbor nodes (if exist) from the vacated neighbor nodes and hence making the technique suitable for {\async}, and additionally,  guaranteeing that the forward phase at a node can finish in $O(1)$ epochs.  
A naive implementation of {\sc P1Tree}, however, needs $O(\Delta)$ bits, which we reduce to $O(\log (k+\Delta))$ through a clever approach.  Although the {\sc P1Tree} construction sounds straightforward, making everything work together needed a careful treatment of several ideas, which we describe in the following. 
%

\vspace{2mm}
\noindent{\bf Techniques.}
While Kshemkalyani {\it et al.} \cite{kshemkalyani2025dispersion} guaranteed $\lceil k/3\rceil$ agents for probing, we guarantee at least $\lceil (k-2)/3\rceil$ agents 
for probing neighbor nodes  (which we call {\em parallel probing}); and show that it is sufficient for $O(1)$ time neighborhood search.
We make $\lceil (k-2)/3\rceil$ agents available by selectively vacating certain nodes of {\sc P1Tree} after an agent settles. This selection typically vacates nodes for which the port-1 neighbor or the port-1 neighbor of a port-1 neighbor is not vacated. 
%
While having previously settled agents helps in $O(1)$ time parallel probing, guaranteeing their availability and not using oscillation create five major challenges Q1--Q5 below. 
We need some notations. Each agent has state from \emph{settled, unsettled, settledScout}. Initially, all $k$ agents are {\em unsettled}, and they travel with the DFShead. At every new node, one agent settles. For a node $v$, we denote the settled agent at that node by $\psi(v)$. The settled agent $\psi(v)$ may not always remain at $v$.


 \begin{itemize}
\item [{\bf Q1.}] {\bf How to run DFS?} 
We run DFS such that it constructs a {\sc P1Tree}, which primarily consists of edges containing port number 1 at their either or both ends (which we call {\em port-1-incident} edges). Therefore, the DFS at a node prioritizes visiting empty neighbors reached via port-1-incident edges.  This priority may create a cycle when  a port-1-incident edge takes the DFS to a node which is already part of a {\sc P1Tree} built so far.  We avoid such cycles by adding an edge which is a non-port-1-incident edge. Additionally, we guarantee that the DFS will never add two consecutive non-port-1-incident edges. It is easy to see that any {\sc P1Tree} has  at least $\lceil n/2\rceil$ port-1-incident edges on it and at most $n-1-\lceil n/2\rceil$ non-port-1-incident edges.   
 
    \item [{\bf Q2.}] {\bf Which {\sc P1Tree} nodes to leave vacant?} 
     We guarantee that we can leave at least $\lfloor l/3\rfloor$ nodes of {\sc P1Tree} of size $l$ vacant. The settled agents at these  $\lfloor l/3\rfloor$ vacant nodes travel with the DFShead and help with parallel probing until DFS ends. The challenge is how to meet such requirement. The general rule of thumb for this decision is as follows. Consider a {\sc P1Tree} node $v$. If $v$ has a port-1-neighbor or port-1 neighbor of port-1 neighbor that is occupied (is not vacant/empty), leave $v$ vacant and collect the agent $\psi(v)$ as a scout (we call $v$ a {\em vacated} node). 

     \item [{\bf Q3.}] {\bf How and where to keep information about vacated nodes of the {\sc P1Tree}?} There are two options: (i) Store the information of vacant node at its occupied port-1 neighbor or (ii) Store information about the port-1 neighbor at the agent $\psi(v)$. Option (i) is problematic since port-1 neighbor of $v$ may also be the port-1 neighbor of multiple nodes and hence the memory need becomes at least $O(\Delta)$ bits. We use Option (ii) such that each agent only keeps track of one port-1 neighbor, using $O(\log (k+\Delta))$ bits.  
    $\psi(v)$ stores the ID of port-1 neighbor, and the port at port-1 neighbor, so that later in parallel probing, $v$ can be correctly identified as vacated.
    
    \item [{\bf Q4.}] {\bf How to successfully run DFS despite some of the {\sc P1Tree} nodes vacant?} 
    Suppose a node $v\in G$. If $v$ is in {\sc P1Tree}, it is either occupied or vacated. If $v$ is not in {\sc P1Tree}, it is empty.
    Suppose a scout agent $a_s$ doing parallel probing from $x$ reaches $x$'s neighbor node $y$. If $a_s$ finds $y$ empty, it checks the port-1 neighbor $z$ of $y$. If $z$ is occupied, $a_s$ returns to $x$. If $z$ is not occupied, it visits port-1 neighbor $w$ of $z$. After visiting $w$, $a_s$ returns to $x$. 
    Due to our strategy of choosing vacant nodes, if $y$ is not empty, then $z$ or $w$ (or both) must be occupied. Even if $y$ and $z$ are vacant, the scout agent can always determine the settled agent at $y$ and $z$ by checking if such agents exist in the scout pool at $x$. Notice that scout $a_s$ visits a (at most) 3 hop neighbor of $w$ in parallel probe starting from $w$, and hence each parallel probe finishes in $6=O(1)$ epochs. 
    The probe needs to search at most $k-2$ ports (excluding parent port) at a non-root node, and the DFShead has at least one unsettled agent. 
    Thus $\lceil (k-2)/3\rceil$ scouts, probing at a node finishes searching in 3 iterations taking at most $18 = O(1)$ epochs. Furthermore, when the degree is more than $k-2$, all unsettled agents can settle by finding empty neighbors with one instance of parallel probing.
    
    \item [{\bf Q5.}] {\bf How to return scout agents to the vacated nodes of the {\sc P1Tree} after DFS finishes?} After having $k$ nodes in {\sc P1Tree}, the DFS finishes. Notice that each scout is associated with a node of {\sc P1Tree}. We ask each scout to carry information about parent/child/sibling details (both ID and port). This information allows the scouts to re-traverse {\sc P1Tree} in post-order of DFS and settle at their associated node when reached. 
    We prove that this re-traversal process finish in $O(k)$ time with $O(\log (k+\Delta))$ memory at each scout.     
\end{itemize}

\noindent{\bf Handling general initial configurations.}
So far we discussed techniques to achieve $O(k)$ time complexity for rooted initial configurations. In general initial configurations, 
there will be $\ell$ DFSs initiated from $\ell$ multiplicity nodes ($\ell$ not known). Each DFS follows the approach as in the rooted case. Let a node has $k_1$ agents running DFS $D_1$. We show that $D_1$ finishes in $O(k_1)$ epochs if $D_1$ does not meet any other DFS, say $D_2$.
In case of a meeting, we develop an approach that handles the meeting of two DFSs $D_1$ and $D_2$ with overhead the size of the larger DFS between the two. In other words, $k_1+k_2$ agents that belong to $D_1$ and $D_2$ disperse in $O(k_1+k_2)$ epochs. If a meeting with the third DFS $D_3$  occurs, we show that it is handled with time complexity $O(k_1+k_2+k_3)$ epochs. Therefore, the worst-case time complexity starting from any $\ell$ multiplicity  nodes becomes $O(k)$. 
Specifically, to achieve this runtime, we extend the {\em size-based subsumption} technique of Kshemkalyani and Sharma \cite{KshemkalyaniOPODIS21} which was also used in the state-of-the-art result of Kshemkalyani {\it et al.} \cite{kshemkalyani2025dispersion}.  The subsumption technique works as follows. 
Suppose DFS $D_1$ meets DFS $D_2$ at node $w$ (notice that $w$ belongs to $D_2$). Let $|D_i|$ denote the number of agents settled from DFS $D_i$ (i.e., the number of nodes in {\sc P1Tree} $T_{D_i}$ built by $D_i$ so far). $D_1$ subsumes $D_2$ if and only if $|D_2|<|D_1|$, otherwise $D_2$ subsumes $D_1$. The agents settled from subsumed DFS (as well as scouts) are collected and given to the subsuming DFS to continue with its DFS, which essentially means that the subsumed DFS does not exist anymore. This subsumption technique guarantees that one DFS out of $\ell'$ met DFSs (from $\ell$ nodes) always remains subsuming and grows monotonically until all agents settle forming a single {\sc P1Tree}.

\vspace{2mm}
\noindent{\bf Related work.}
Table~\ref{table:comparision} reviews the state-of-the-art time– and memory-efficient solutions for {\dis}.  
The current best algorithms are due to Kshemkalyani {\it et al.}~\cite{kshemkalyani2025dispersion}, who give an optimal $O(k)$-round solution with $O(\Delta+\log k)$ memory in the synchronous model and an $O(k\log k)$-epoch solution with $O(\log(k+\Delta))$ memory in the asynchronous model.  
Our contribution attains the same optimal $O(k)$ bound in the asynchronous model while matching the $O(\log(k+\Delta))$ memory footprint.

These advances extend a long research line~\cite{Augustine:2018,BKM24,BKM25,ChandKMS23,DasCALDAM21,GorainSSS22,ItalianoPS22,KaurD2D23,KshemkalyaniICDCN19,KshemkalyaniALGOSENSORS19,KshemkalyaniWALCOM20,KshemkalyaniICDCS20,KshemkalyaniOPODIS21,tamc19,SaxenaK025,Saxena025,ShintakuSKM20}.  
Most papers study the fault-free setting; notable exceptions handle Byzantine agents~\cite{BKM25,MollaIPDPS21,MollaMM21} or crash faults~\cite{BKM24,BKM25,ChandKMS23,Pattanayak-WDALFR20}.  
The prevailing communication model is {\em local} (only co-located agents interact); the sole {\em global} variant appears in~\cite{KshemkalyaniICDCN20}.  
While nearly all works assume static graphs, dynamic topologies were explored in~\cite{KshemkalyaniICDCS20,SaxenaK025,Saxena025}.  
Deterministic algorithms dominate; randomization was used mainly to optimize memory~\cite{DasCALDAM21,tamc19}; however several deterministic solutions attained the same memory complexity.  
Variants with restricted communication~\cite{GorainSSS22} or constrained final configurations (e.g., independent-set dispersion~\cite{KaurD2D23}) have also been considered.

Although the general problem is tackled on arbitrary graphs, specialized studies address grids~\cite{BKM24,BKM25,KshemkalyaniWALCOM20}, rings~\cite{Augustine:2018,MollaMM21}, and trees~\cite{Augustine:2018,KshemkalyaniICDCN20}.  
Some algorithms further assume a priori knowledge of parameters such as $n$, $m$, $\Delta$, or $k$~\cite{ChandKMS23,KshemkalyaniALGOSENSORS19}.

{\dis} is closely related to graph exploration~\cite{Bampas:2009,Cohen:2008,Dereniowski:2015,Fraigniaud:2005,MencPU17}, scattering on rings and grids~\cite{ElorB11,Shibata:2016,Barriere2009,Das16,Poudel18}, and load balancing~\cite{Cybenko:1989,Subramanian:1994}.  
Most recently, a {\dis} routine has been leveraged to elect a leader and to compute graph-level structures such as MST, MIS, and MDS~\cite{KshemkalyaniKMS24}.

\onlyLong{%
Exploration with limited memory has been intensely studied as well.  
Fraigniaud {\it et al.}~\cite{Fraigniaud:2005} showed that a single agent with $\Theta(D\log\Delta)$ bits can explore any anonymous graph in $O(\Delta^{D+1})$ time.  
Allowing nodes to store a few bits, Cohen {\it et al.}~\cite{Cohen:2008} proposed two trade-offs:  
(i) $O(1)$ bits at the agent and $2$ bits per node, and  
(ii) $O(\log\Delta)$ bits at the agent and $1$ bit per node; both algorithms run in $O(m)$ time after an $O(mD)$ preprocessing phase.  
The time/agent-count trade-off is analyzed in~\cite{MencPU17}, and rotor-router exploration bounds are given in~\cite{Bampas:2009}.  
Collective exploration of trees appears in~\cite{FraigniaudGKP06}.  
To our knowledge, we are the first to handle both crash and Byzantine faults for {\dis} on arbitrary graphs within the local communication model.}

\vspace{2mm}
\noindent{\bf Roadmap.} 
The model  details and preliminaries are given in Section \ref{section:model}.  
In Section~\ref{sec:spanning-tree}, we build some techniques, including a {\sc P1Tree} construction via a DFS traversal. We then discuss in Section \ref{sec:agentp1tree} how to construct a {\sc P1Tree} with mobile agents, which is crucial for our algorithm in {\async}. 
The $O(k)$-time {\async} algorithm for the rooted case is described in Section \ref{sec:rooted}. An extension to the general case keeping $O(k)$ time is described in Section \ref{sec:general}.  
Finally, Section \ref{sec:conclusion} concludes the paper with a  short discussion. Pseudocodes and some proofs are deferred to Appendix. 

\section{Model and Preliminaries}
\label{section:model}

\noindent{\bf Graph.} We consider a simple, undirected, connected graph $G = (V, E)$, where $n = |V|$ is the number of nodes and $m = |E|$ is the number of edges. For any node $v \in V$, let $N(v)$ denote the set of its neighbors and let $\delta_v = |N(v)|$ denote its degree. The maximum degree of the graph is $\Delta = \max_{v \in V} \delta_v$.
Graph nodes are \emph{anonymous} (i.e., they lack unique identifiers). However, the graph is \emph{port-labeled}: at each node $v$, the incident edges are assigned distinct local labels (port numbers) from $1$ to $\delta_v$, enabling agents at $v$ to distinguish between the outgoing edges.
The port number at $u$ for an edge $\{u, v\}$ is denoted by $p_{uv}$. An edge $\{u, v\}$ is associated with two port numbers: $p_{uv}$ at node $u$ and $p_{vu}$ at node $v$. These port numbers are assigned locally and independently at each endpoint; hence, it is possible that $p_{uv} \ne p_{vu}$, and there is no inherent correlation between port assignments at different nodes. Each node $u\in V$ is memory-less.

\vspace{2mm}
\noindent{\bf Agents.} 
The system comprises $k \leq n$ mobile agents, $A = \{a_1, \dots, a_k\}$. Each agent $a_i$ is endowed with a unique positive integer identifier, $a_i.\mathsf{ID}$, drawn from the range $[1,\, k^{O(1)}]$.
Since agents are assumed to be positioned arbitrarily initially, there may be the case that all $k\leq n$ agents are at the same node, which we denote as {\em rooted initial configuration}. Any initial configuration that is not rooted is denoted as {\em general}. In any general initial configuration, agents are on at least two nodes. A special case of general configuration is a {\em dispersion configuration} in which $k$ agents are on $k$ different nodes.
We consider the {\em local} model in which a agent at a node can only communicate with other agents co-located at that node.

\vspace{2mm}
\noindent{\bf Time cycle.} An agent $a_i$
 could be active at any time. Upon activation, $a_i$ executes the ``Communicate-Compute-Move'' (CCM) cycle as follows. 
 \begin{itemize}
     \item[]
     {\bf Communicate:} Agent $a_i$ positioned at node $u$ can observe the memory of another agent $a_j$ positioned at node $u$. Agent $a_i$ can also observe its own memory.
     \item[] {\bf Compute:} 
 Agent $a_i$ may perform an arbitrary computation using the information observed during the ``Communicate'' portion of that cycle. This includes the determination of a port to use to exit $u$
 and the information to store in the agent $a_j$
 that is at $u$.
     \item[] {\bf Move:} At the end of the cycle, 
$a_i$ writes new information in its memory as well as in the memory of an agent $a_j$
 at 
$u$, and exits $u$  %
using the computed port to reach a neighbor. 
$u$. 
 \end{itemize}

\vspace{1mm}
\noindent{\bf Round, epoch, time, and memory complexity.}
In {\sync}, all agents have common notion of time and activate in discrete intervals called \emph{rounds}. 
In {\async}, agents can have arbitrary activation times and can activate at arbitrary frequency. The restriction is that every agent is active infinitely often, and each cycle finishes in finite time. 
An {\em epoch} is a minimal interval within which each agent finishes at least one CCM cycle \cite{Cord-LandwehrDFHKKKKMHRSWWW11}. Formally, let $t_0$ denote the start time. Epoch $i\geq 1$ is the time interval from $t_i-1$ to $t_i$ where $t_i$ is the first time instant after $t_i-1$ when each agent has finished at least one complete CCM cycle. Therefore, for {\sync}, a round is an epoch. 
We will use the term ``time'' to mean rounds for {\sync} and epochs for {\async}.
Memory complexity is the number of bits stored at any agent over one CCM cycle to the next. The temporary memory needed during the Compute phase is considered free.

\input{0_post}

\input{probe}
\input{11_rooted}
\input{2_general}

\section{Concluding Remarks}\label{sec:conclusion}
We have considered in this paper a fundamental problem of  {\dis} which asks $k\leq n$ mobile agents with limited memory positioned initially arbitrarily on the nodes (memory-less) of an $n$-node port-labeled anonymous graph of maximum degree $\Delta$ to autonomously relocate to the nodes of the graph such that each node hosts no more than one agent.  This problem has been studied extensively recently focusing on the objective of minimizing time and/or memory complexities. A latest state-of-the-art study provided the optimal time complexity of $O(k)$ in the synchronous setting but only able to show $O(k\log k)$ time complexity in the asynchronous setting. We have closed this complexity gap by providing a $O(k)$ time complexity solution in the asynchronous setting. Our solution is obtained through a novel technique of port-one tree we develop in this paper which prioritizes visiting edges with port-1 at (at least) one end-point. Our result is significant since it shows that synchrony assumption is not a requirement for a time-optimal dispersion solution. For the future work, it would be interesting to explore whether our port-one tree technique could be useful in solving other fundamental problems in port-labeled anonymous graphs.

\bibliographystyle{plain}
\bibliography{references} 

\input{example}
\input{tablevariable}
\input{3_pseudocodes}

\end{document}

%% file: fig_initial.tex
\begin{figure}[!t]
\centering
\ifshowtikz
\resizebox{0.45\textwidth}{!}{%
\begin{tikzpicture}[
    scale=3.2,
    every node/.style={circle, draw=black, fill=white, minimum size=12mm, text=black, inner sep=1pt},
    port/.style={circle, fill=white, draw=black, text=black, font=\footnotesize, inner sep=1pt, minimum size=4mm},
    inner_circle/.style={draw=black, minimum size=3.5mm, inner sep=0pt},
    edge/.style={draw=black, thick}
  ]
  \definecolor{c1}{rgb}{1,0,0}      
  \definecolor{c2}{rgb}{0,0.8,0}    
  \definecolor{c3}{rgb}{0,0,1}      
  \definecolor{c4}{rgb}{0,1,1}      
  \definecolor{c5}{rgb}{1,0,1}      
  \definecolor{c6}{rgb}{0.5,0.5,0.5}      
  \definecolor{c7}{rgb}{1,0.65,0}   
  \definecolor{c8}{rgb}{0.54,0.087,0.89} 

  \node (n0)  at (0,0)    {};
  \node (n1)  at (1,1)    {};
  \node (n2)  at (1,-1)   {};
  \node (n3)  at (2, 0.3) {};
  \node (n4)  at (3, 1.5) {};
  \node (n5)  at (3.5, 0) {};
  \node (n6)  at (4.5, 1.2) {};
  \node (n7)  at (2.8, -1) {};
  \node (n8)  at (3.5, -1.8){};
  \node (n9)  at (4.5, -0.8){};

  \foreach \source/\dest/\psrc/\pdst in {
      n0/n1/1/1, n0/n2/2/1, n0/n3/3/1,
      n1/n3/2/2, n2/n3/2/3, n3/n4/4/1,
      n4/n6/2/1, n4/n5/3/1, n5/n6/2/2,
      n5/n7/3/1, n5/n9/4/1, n5/n8/5/1,
      n7/n8/2/2, n9/n8/2/3
  } {
      \draw[edge] (\source) -- (\dest);
      \node[port] at ($(\source)!3mm!(\dest)$) {\psrc};
      \node[port] at ($(\dest)!3mm!(\source)$) {\pdst};
  }
  \node[inner_circle, fill=c1, rectangle] at (n0.center) {};
  \node[inner_circle, fill=c2, signal, signal to=west] at ($(n3.center)+(-0.08, 0.08)$) {};
  \node[inner_circle, fill=c5, star] at ($(n3.center)+(0.08, 0.08)$) {};
  \node[inner_circle, fill=c3, cylinder, minimum height=2pt] at ($(n3.center)+(-0.08, -0.08)$) {};
  \node[inner_circle, fill=c4, regular polygon, regular polygon sides = 3] at ($(n3.center)+(0.08, -0.08)$) {};
  \node[inner_circle, fill=c7, diamond] at ($(n5.center)+(0.0, 0.08)$) {};
  \node[inner_circle, fill=c8, signal, signal to=east] at ($(n5.center)+(-0.08, -0.08)$) {};
  \node[inner_circle, fill=c6, circle] at ($(n5.center)+(0.08, -0.08)$) {};

\end{tikzpicture}%
}
\hspace{1em} 
\resizebox{0.45\textwidth}{!}{%
\begin{tikzpicture}[
    scale=3,
    every node/.style={circle, draw=black, fill=white, minimum size=12mm, text=black, inner sep=1pt},
    port/.style={circle, fill=white, draw=black, text=black, font=\footnotesize, inner sep=1pt, minimum size=4mm},
    inner_circle/.style={draw=black, minimum size=3.5mm, inner sep=0pt},
    edge/.style={draw=black, thick}
  ]
  \definecolor{c1}{rgb}{1,0,0}      
  \definecolor{c2}{rgb}{0,0.8,0}    
  \definecolor{c3}{rgb}{0,0,1}      
  \definecolor{c4}{rgb}{0,1,1}      
  \definecolor{c5}{rgb}{1,0,1}      
  \definecolor{c6}{rgb}{0.5,0.5,0.5}      
  \definecolor{c7}{rgb}{1,0.65,0}   
  \definecolor{c8}{rgb}{0.54,0.087,0.89} 

  \node (n0)  at (0,0)    {};
  \node (n1)  at (1,1)    {};
  \node (n2)  at (1,-1)   {};
  \node (n3)  at (2, 0.3) {};
  \node (n4)  at (3, 1.5) {};
  \node (n5)  at (3.5, 0) {};
  \node (n6)  at (4.5, 1.2) {};
  \node (n7)  at (2.8, -1) {};
  \node (n8)  at (3.5, -1.8){};
  \node (n9)  at (4.5, -0.8){};

   \foreach \source/\dest/\psrc/\pdst in {
      n0/n1/1/1, n0/n2/2/1, n0/n3/3/1,
      n1/n3/2/2, n2/n3/2/3, n3/n4/4/1,
      n4/n6/2/1, n4/n5/3/1, n5/n6/2/2,
      n5/n7/3/1, n5/n9/4/1, n5/n8/5/1,
      n7/n8/2/2, n9/n8/2/3
  } {
      \draw[edge] (\source) -- (\dest);
      \node[port] at ($(\source)!3mm!(\dest)$) {\psrc};
      \node[port] at ($(\dest)!3mm!(\source)$) {\pdst};
  }

  \node[inner_circle, fill=c1, rectangle] at (n0.center) {};
  \node[inner_circle, fill=c2, signal, signal to=west] at (n1.center) {};
  \node[inner_circle, fill=c3, cylinder, minimum height=2pt] at (n2.center) {};
  \node[inner_circle, fill=c4, regular polygon, regular polygon sides = 3] at (n3.center) {};
  \node[inner_circle, fill=c5, star] at (n4.center) {};
  \node[inner_circle, fill=c6, circle] at (n5.center) {};
  \node[inner_circle, fill=c7, diamond] at (n6.center) {};
  \node[inner_circle, fill=c8, signal, signal to=east] at (n7.center) {};

\end{tikzpicture}%
}
\fi
\caption{{\dis} of 8 mobile agents in a 10-node graph. On the left, 8 agents are initially located at three different nodes. On the right, agents are dispersed to occupy one node each.}
\label{fig:dispersion}
\end{figure}
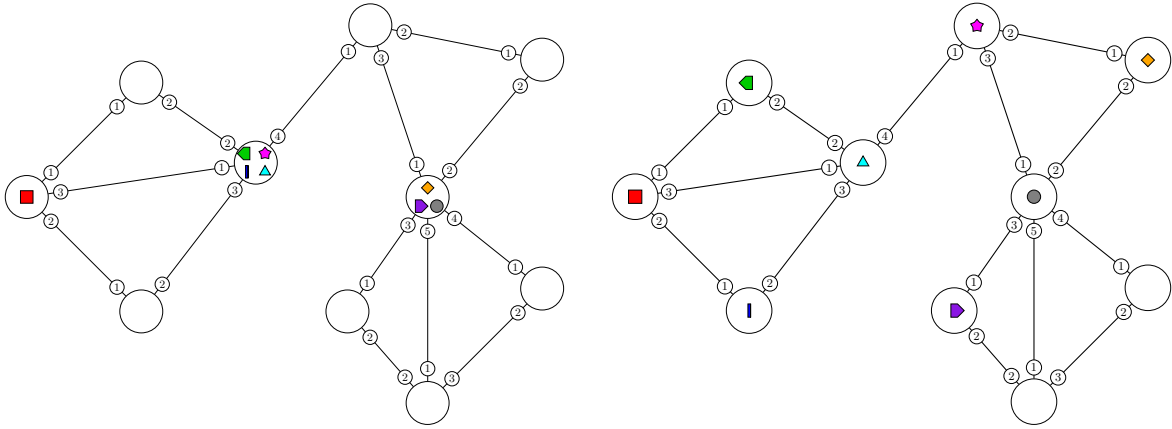

%% file: 0_post.tex
\input{terminology.tex}
\section{Port-One Tree and its Construction}
\label{sec:spanning-tree}

In this section, we first define port-one tree ({\sc P1Tree}) and discuss its centralized and distributed construction. Then we describe a method of selecting nodes to be ``vacated'' on the constructed {\sc P1Tree}. 
We finally discuss parallel probing technique. 

\vspace{2mm}
\noindent{\bf Port-One Tree ({\sc P1Tree}).}
Intuitively, every vertex in a {\sc P1Tree} $\mathcal{T}$ is incident to at least one tree edge carrying port~$1$ at one of its end-points. Formally,
\begin{definition}[Port-One Tree (\textsc{P1Tree})]
\label{definition:p1tree}
Let $G=(V,E)$ be an anonymous port-labeled graph.
A tree $\mathcal{T}\subseteq E$ is a {\sc P1Tree} 
if each vertex $v\in \mathcal{T}$ has  (at least) one incident edge leading to, say node w, such that  edge $\{v,w\}\in \mathcal{T}$ and
 $\texttt{type}(\{v,w\}) \in \{\texttt{tp1}, \texttt{t11}, \texttt{t1q}\}.$
%
%
\end{definition}

There may be multiple trees $\mathcal{T}$ that satisfy Definition \ref{definition:p1tree} of a {\sc P1Tree}. Let that set of trees $\mathcal{T}$ be denoted as $\mathbf{T}$. In this paper, we are interested in constructing a {\sc P1Tree} $\mathcal{T}
\in \mathbf{T}$. 
We will prove later there exists at least a  \textsc{P1Tree} $\mathcal{T}\in \mathbf{T}$ for any graph $G$. 
%
%
\input{fig_examplepost}
%
A \textsc{P1Tree} $\mathcal{T}$ is shown in Fig.~\ref{fig:examplepost}. Notice that, a \textsc{P1Tree} $\mathcal{T}$ may not contain all the edges of type \texttt{tp1}, \texttt{t11}, or \texttt{t1q} because doing so may create cycles. We avoid cycles by adding edges of type \texttt{tpq}. 

Since each node $v\in G$ has degree $\delta_v\geq 1$, Observation \ref{obs:t11t1q} follows immediately.
\begin{observation}\label{obs:t11t1q}
    Each node in a port-labeled graph $G$ has an edge of type {\em \texttt{t11} \emph{or} \texttt{t1q}}.
\end{observation}

\begin{lemmarep}\label{lemma:p1tree}
    For any port-labeled graph $G$, there exists at least one {\sc P1Tree} $\mathcal{T}$. 
\end{lemmarep}
\begin{proof}
    For the sake of contradiction, assume that no {\sc P1Tree} exists for a connected graph $G = (V,E)$. Consider any spanning tree $\mathcal{T}' \subseteq E$ of $G$. Since $\mathcal{T}'$ is not a {\sc P1Tree}, there is at least one vertex $v \in V$ that does not satisfy port-one property. If $\delta_v = 1$, then the edge connecting $v$ must have port 1, hence $\delta_v > 1$. 
    For such a node $v$, consider the port-one neighbor $w$ of $v$. Consider node $z$ such that $\{v,z\} \in \mathcal{T}'$ and the path from $v$ to $w$ contains $z$. Now, we substitute the edge $\{v,z\}$ with edge $\{v,w\}$. The tree preserves all the properties, and additionally now $v$ satisfies the port-one property. 
    Repeating this process for every node not satisfying the port-one property, we arrive at a tree $\mathcal{T}$, where all the nodes satisfy port-one property and hence the assumption is false.
\end{proof}
The proof is deferred to the Appendix.


\vspace{2mm}
\noindent{\bf Centralized construction.}
The pseudocode for the centralized construction 
is given in Algorithm \ref{alg:portonecentral} in Appendix \ref{app:centralized} (we name our algorithm \texttt{Centralized\_P1Tree()}). Initially,  $\mathcal{T}$ is empty, i.e., $\mathcal{T}\leftarrow \emptyset$.  Let $\mathcal{C}_i$ be a tree component initially empty, i.e., $\mathcal{C}_i\leftarrow \emptyset$. The goal is construct $\kappa\geq 1$ ($\kappa$ not known) tree components $\mathcal{C}_1,\ldots,\mathcal{C}_\kappa$ such that $\mathcal{C}_1\cap\ldots\cap \mathcal{C}_\kappa=\emptyset$ and $\mathcal{C}_1\cup\ldots\cup\mathcal{C}_\kappa=V$.  The $\kappa$ components are then connected via $\kappa-1$ edges to obtain $\mathcal{T}$. 
The algorithm starts from an arbitrary node $v\in G$. It adds in $\mathcal{C}$ (initially empty) all the incident edges of $v$ of types $\texttt{tp1},\texttt{t11}$, and $\texttt{t1q}$ one by one in the priority order. While doing so, each neighboring node is added in a stack. If all such edges at $v$ are exhausted, then it goes to a neighbor (top of stack) and repeats the procedure. 
If adding an edge $\{w,z\}$ of type $\texttt{tp1},\texttt{t11}$, or $\texttt{t1q}$ at node $w$ creates a cycle, we set $\mathcal{C}_1\leftarrow \mathcal{C}$. 
The algorithm then continues constructing $\mathcal{C}$ on $G\backslash \mathcal{C}_1$ until a cycle is detected. It then sets $\mathcal{C}_2\leftarrow \mathcal{C}$. The component construction stops as soon as stack goes empty, which also means that $G\backslash \{\mathcal{C}_1\cup\ldots\cup\mathcal{C}_\kappa\}=\emptyset.$ 
We then include all these components in $\mathcal{T}$. We then connect these $\kappa$ tree components adding $k-1$ edges as follows: 
we sort the edges of $G$ that are not yet considered to construct $C_1,\ldots,C_\kappa$ in lexicographical order. These edges must be of type \texttt{tpq}. We then add these edges to $\mathcal{T}$ as long as adding the edge does not create cycle. The algorithm terminates, when no more edges can be added. This also means that there is only one component containing all the nodes of $G$.

\begin{lemmarep}\label{lemma:centralized}
    Given a port-labeled graph $G$, Algorithm \ref{alg:portonecentral} correctly constructs a {\sc P1Tree} $\mathcal{T}$. 
\end{lemmarep}
\begin{proof}
Suppose, to the contrary, that the subgraph ~$\mathcal{T}$ is not a \textsc{P1Tree}.  Then either
\begin{enumerate}[label=(\arabic*)]
  \item $\mathcal{T}$ is not a tree, or
  \item some vertex~$x\in\mathcal{T}$ has no incident edge of type \texttt{tp1}, \texttt{t11}, or \texttt{t1q}.
\end{enumerate}
But in Algorithm~\ref{alg:portonecentral} (\texttt{Centralized\_P1Tree()}) every inserted edge either
\begin{itemize}[nosep]
  \item joins an \textbf{unvisited} vertex (so it cannot create a cycle) or
  \item is a \texttt{tpq} edge that connects two previously disjoint port-one components $C_i,C_j$.
\end{itemize}
In either of the cases, no cycle is created.
Hence $\mathcal{T}$ is a tree, contradicting (1).  
Furthermore, every node~$v\in G$ has at least one port-1 edge (\texttt{t11} or \texttt{t1q}). 
Algorithm~\ref{alg:portonecentral} would fail to add such an edge for node $v$ if it creates a cycle. In that case, there already exists an edge that connects to $v$. Since the port~1 edge at $v$ has not been considered yet, the edge that connects $v$ must contain a port~1 at one of its end-points. Thus $v$ still ends up with some port~1 incident edge in~$\mathcal{T}$, contradicting (2).  Therefore, $\mathcal{T}$ must be a {\sc P1Tree}. 
\end{proof}
The proof is deferred to the Appendix.

\vspace{2mm}
\noindent{\bf DFS-based construction.}\label{subsec:dfspost}
DFS explores $G$ through forward and backtrack phases, switching between them as needed.  The forward phase takes it as deeply as possible along each branch and the backtrack phase helps finding nodes from which forward phase can continue again. Let {\em DFShead} be the node where DFS is currently performing forward or backtrack phase.  Initially, the starting node of DFS acts as the DFShead. 

We need some terminologies. We say that each node $v\in G$  has two states:
\begin{itemize}
    \item \textsc{empty}: $v$ has not been visited by the DFShead yet.
    \item \textsc{occupied}: $v$ has been visited by the DFShead already. 
\end{itemize}
We denote by {\em parent edge} of  $v\in \mathcal{T}$ the edge in $\mathcal{T}$ from which DFS first visited $v$ doing forward phase. 
We call the associated type of the edge as {\em parent edge type}. 
Moreover, we categorize each node $v\in G$ into the following four types:  

\begin{itemize}
    \item \textbf{unvisited}: $v$ is not yet visited by the DFShead.
    \item \textbf{fullyVisited}: $v$ is visited by the DFShead and $v$ has no empty neighbors. 
    \item \textbf{partiallyVisited}: $v$ has parent edge of type \texttt{tpq} and each empty neighbor is reached by an edge of type \texttt{tpq}.
    \item \textbf{visited}: $v$ is visited by the DFShead and $v$ is not \textbf{partiallyVisited} or \textbf{fullyVisited}.
\end{itemize}
We overload the $\texttt{type}()$ function to indicate the type of a node $u$ in $\{$\textbf{unvisited}, \textbf{partiallyVisited}, \textbf{fullyVisited}, \textbf{visited}$\}$.

\vspace{1em}
\noindent{\bf Algorithm.}
We now describe the algorithm to construct a {\sc P1Tree} $\mathcal{T}\in \mathbf{T}$ (we call our algorithm \texttt{DFS\_P1Tree()}; the pseudocode is given in Algorithm \ref{alg:portoneDFS} in Appendix \ref{app:distributed}).  \texttt{DFS\_P1Tree()} prioritizes the edge types at any node $v\in V$ in the following order: 
$\texttt{tp1}\succ\texttt{t11}\sim\texttt{t1q}\succ\texttt{tpq}$ 
(for multiple edges of same type at $v$, they are prioritized in the increasing order of the port number at $v$). 

Suppose DFShead reaches node $u$. We determine the type of node $u$ by checking $N(u)$. Based on the type of $u$, the DFShead does the following:
\begin{itemize}[leftmargin=3em]
    \item[(D0)] all nodes are initially \textbf{unvisited}.
    \item[(D1)] if $u$ is \textbf{fullyVisited}, then DFShead backtracks to parent of $u$.
    \item[(D2)] if $u$ is \textbf{visited}, then DFShead continues along the highest priority edge to an empty neighbor of $u$.
\end{itemize}
\texttt{DFS\_P1Tree()} terminates when all nodes of $G$ become \textbf{fullyVisited}.
The rules (D0), (D1), and (D2) are analogous to the standard DFS traversal with a major change that introduces 
\textbf{partiallyVisited} node types to 
ensure that each node of $\mathcal{T}$ satisfy the {\sc P1Tree} property. In particular, \texttt{DFS\_P1Tree()} 
has following additional rules:
\begin{itemize}[leftmargin=3em]
    \item[(D3)] if $u$ is \textbf{partiallyVisited}, then DFShead backtracks to parent of $u$.
    \item[(D4)] a \textbf{partiallyVisited} node $u$ has state \textsc{empty} when DFShead is at a node $w$ such that $\texttt{type}(\{w,u\} \in \{\texttt{tp1},\texttt{t11}\}$; otherwise, \textsc{occupied}.
\end{itemize}
\texttt{DFS\_P1Tree()} converts a \textbf{partiallyVisited} node $u$ to a \textbf{visited} node when DFShead visits $u$ from $w$ such that $p_{uw}=1$ (i.e., $w$ is the port-1 neighbor of $u$). Additionally, the parent edge of $u$ now {\em swapped} to make $w$ the parent of $u$ in $\mathcal{T}$. 
We prove later that this parent swap does not create a cycle. 
We call this process the  \emph{reconfiguration} of a \textbf{partiallyVisited} node.  

Fig.~\ref{fig:reconfig} illustrates these ideas. As shown in Fig.~\ref{fig:post}, the DFS backtracks at \textcircled{\scriptsize 6} marking \textcircled{\scriptsize 6} \textbf{partiallyVisited}, since the parent edge $\{$\textcircled{\scriptsize 6},\textcircled{\scriptsize 5}$\}$ is of type \texttt{tpq} and its only \textbf{unvisited} neighbor \textcircled{\scriptsize 9} is connected via a \texttt{tpq} edge $\{$\textcircled{\scriptsize 6},\textcircled{\scriptsize 9}$\}$. As shown in Fig.~\ref{fig:reconfiguredpost}, when DFS reaches again to \textcircled{\scriptsize 6} in the forward phase via edge $\{$\textcircled{\scriptsize 0},\textcircled{\scriptsize 6}$\}$, parent  of \textcircled{\scriptsize 6} is now swapped to \textcircled{\scriptsize 0} (i.e.,  edge $\{$\textcircled{\scriptsize 6},\textcircled{\scriptsize 5}$\}$ is removed and edge $\{$\textcircled{\scriptsize 0},\textcircled{\scriptsize 6}$\}$ is added), which makes \textcircled{\scriptsize 6} \textbf{visited}. 

\input{fig_post.tex}


\begin{theorem}\label{thm:portoneDFS}
    \texttt{DFS\_P1Tree()} produces a {\sc P1Tree} $\mathcal{T}$ of a port-labeled graph $G$.
\end{theorem}
\begin{proof}
We prove the theorem via three claims.

\textbf{Claim 1.}
{\em Every vertex of $G$ is popped from the stack at most twice and is eventually marked \textbf{fullyVisited}:} A vertex is pushed onto the stack when it is discovered for the first time, i.e., as an \textbf{unvisited} vertex. 
A \textbf{partiallyVisited} vertex may be pushed a second time, but only when it is visited from its port-1 (rule (D4)). At this second push, the vertex immediately becomes \textbf{visited} and, popped only when it is declared \textbf{fullyVisited}. Thus, every vertex is popped at most twice, so the stack becomes empty.

\textbf{Claim 2.}
{\em At every step, $\mathcal{T}$ is a tree; when the algorithm halts, it is a spanning tree of $G$:} Edges are inserted into $\mathcal{T}$ in two ways:

\begin{itemize}
    \item When DFS follows an edge to an \textbf{unvisited} vertex, $e=[u,p_{uv},p_{vu},v]$ joins $u$ (already in the tree) to a new vertex $v$. No cycle is created.
    \item During a reconfiguration, the current parent edge $e_{\uparrow}=[w,p_{wu},p_{uw},u]$ (necessarily of type \texttt{tpq}) is removed and replaced with a port-1 edge $e=[u,p_{uv'},p_{v'u},v']$. Since the swap deletes the unique path between $u$ and $w$ before adding $e$, no cycle is produced and the tree built so far remains connected.
\end{itemize}

Thus, $\mathcal{T}$ is always a tree. Since DFS starts at the root $v_0$ and eventually discovers every vertex (by Claim~1), the final tree spans $V$.

\textbf{Claim 3.}
{\em When the algorithm terminates, every vertex of $\mathcal{T}$ is incident to at least one edge of type \texttt{tp1}, \texttt{t11}, or \texttt{t1q}:} Consider a vertex $x$ when it is popped for the last time.

\begin{itemize}
    \item $x$ was discovered through an edge of type \texttt{tp1}, \texttt{t11}, or \texttt{t1q}. That edge remains in $\mathcal{T}$ (never swapped out), so the property holds for $x$.
    \item $x$ was first discovered through a \texttt{tpq} edge. By definition, $x$ is immediately marked \textbf{partiallyVisited} and DFS backtracks (rule (D3)). Observation~\ref{obs:t11t1q} guarantees that $x$ has a neighbor $y$ such that the edge $\{x,y\}$ is port-1 at one-end. Due to the edge priority ($\texttt{tp1} \succ \texttt{t11} \sim \texttt{t1q} \succ \texttt{tpq}$, Line~\ref{alg:portoneDFS:priority}), DFS will eventually reach $y$ and then revisit $x$ via this port-1 edge. The reconfiguration at that moment (rule (D4)) replaces the old \texttt{tpq} parent by a port-1 edge, after which $x$ is \textbf{visited}. From then on, the incident port-1 edge is never removed, so the property holds when $x$ finally becomes \textbf{fullyVisited}.
\end{itemize}

Thus, every vertex ends with an incident port-1 tree edge.
The algorithm terminates when all vertices are \textbf{fullyVisited} (Claim~1). By Claim~2, $\mathcal{T}$ is a spanning tree. By Claim~3, every vertex in $\mathcal{T}$ has an incident edge of type \texttt{tp1}, \texttt{t11}, or \texttt{t1q}. Thus, $\mathcal{T}$ is a {\sc P1Tree}.
\end{proof}

\section{Constructing {\sc P1Tree} with Agents}\label{sec:agentp1tree}
Now, we describe how Algorithm~\ref{alg:portoneDFS} (\texttt{DFS\_P1Tree()}) can be executed by agents. We first describe a straightforward (non-optimal) construction of the {\sc P1Tree}, and then show the utilization of certain structural properties of {\sc P1Tree} to construct it in optimal time. 

\vspace{2mm}
\noindent\textbf{High-level overview.}
Suppose agents are initially located at a node $v_0$. The highest ID agent settles at $v_0$, which keeps track of the state and type of node $v_0$. Now the agents perform a neighborhood search to determine the edge types and choose the highest priority edge leading to an empty node. 
This edge is chosen as the next edge to be included in the DFS tree. The neighborhood search means that the neighbors of the current node is visited one-by-one, and the settled agents in the neighborhood nodes indicate the state and type of the nodes.
On reaching an unvisited node, a new agent (highest ID among the unsettled agents) settles. It sets its parent to the node where the DFShead arrived from. The result of neighborhood search determines the type and state of the node. 
When no empty neighbor is available, the DFShead travels back to the parent node of the current node. Analogous to Algorithm~\ref{alg:portoneDFS} (\texttt{DFS\_P1Tree()}), DFShead at $x$ moves in the forward phase when the neighbor $y$ is \textbf{partiallyVisited} and $p_{yx} = 1$. The settled agent at $y$, changes its type, and parent.  

\vspace{2mm}
\noindent\textbf{Preliminaries.}
Consider there are $n$ agents initially located at a node $v_0$ of the $n$-node graph $G$. We call $v_0$ the root node. We draw analogy from the node states and types to describe the agent states and variables. Then we show how we maintain the same information using agents.
Initially, all agents are in state \emph{unsettled}.
We say an agent $a$ is settled at node $x$, by setting the state of agent $a$ to \emph{settled}. 
Recall that the nodes are anonymous and hence there is no identifier associated with them. We associate the ID of the agent settled at a node as a proxy for the ID of the node. We represent the settled agent at node $x$ as $\psi(x)$.
Similarly, an agent $a$ at $x$ can identify the type of the edge $e_{xy}$ after traversing it, since it can only know the port number $p_{yx}$ after reaching $y$.
When agent $a$ gets settled at node $x$, it also stores the type of the node $x$ in variable $a.\textsf{nodeType}$. Initially, $a.\textsf{nodeType}$ is \textbf{unvisited}, but the agent must obtain its correct node type by doing a neighborhood search. 

\vspace{2mm}
\noindent\textbf{Initialization.}
The DFShead is initially located at $v_0$.
The highest ID agent among all the agents present settles at $v_0$. Let $a_0$ be that agent. We say $\psi(v_0) = a_0$. The agent $a_0$ changes its state $a_0.\textsf{state}$ to \emph{settled}.
The agents construct the tree by keeping track of their parent node. 
Initially, $a_0.\textsf{parentID} = \bot$, $a_0.\textsf{parentPort} = \bot$, and  $a_0.\textsf{portAtParent} = \bot$.

\vspace{2mm}
\noindent\textbf{Neighborhood search.}
Similar to Algorithm~\ref{alg:portoneDFS} (\texttt{DFS\_P1Tree()}), the DFShead must find the next edge to travel according to the priority order. 
To determine the next edge in priority order, the agents at $x$ traverse the ports in the increasing order. 
They set their phase to \emph{probe}, and traverse an edge to determine the state and type. On traversing an edge corresponding to port $p_{xy}$, the agents determine the result as a 4-tuple $\langle p_{xy}, \texttt{type}(\{x,y\}), \texttt{type}(y), \psi(y) \rangle$. Notice that, all of the constituents of the 4-tuple can be determined locally, since when $y$ is \textbf{unvisited}, then $\psi(y) = \bot$. After determining the results for all the ports incident at $x$, the agents choose the next edge to traverse accordingly.

\vspace{2mm}
\noindent\textbf{Handling partiallyVisited nodes.} 
In Algorithm~\ref{alg:portoneDFS} (\texttt{DFS\_P1Tree()}), a node is \textbf{partiallyVisited} when its parent edge type is \texttt{tpq}, and there are empty neighbors connected by \texttt{tpq} edges. The agents can determine this locally after performing a neighborhood search and then a node $w$ is marked \textbf{partiallyVisited} by retaining this information at $\psi(w).\textsf{nodeType}$. Then the unsettled agents leave via $\psi(w).\textsf{parentPort}$ to reach the parent of $w$ in the tree.

On the other hand, when the highest priority result during neighborhood search is the tuple $\langle p_{xw}, \texttt{tp1}/\texttt{t11}, \textbf{partiallyVisited}, \psi(w) \rangle$, the agents move to $w$ via $p_{wx} = 1$, and thus $\psi(w)$ updates its parameters to $\psi(w).\textsf{parent} = (\psi(x).\textsf{ID}, p_{xw})$, $\psi(w).\textsf{nodeType} = \textbf{visited}$ and $\psi(w).\textsf{parentPort} = 1$. 

\vspace{2mm}
\noindent\textbf{Termination.}
Similar to Algorithm~\ref{alg:portoneDFS} (\texttt{DFS\_P1Tree()}), the termination happens when all nodes are of type \textbf{fullyVisited}. The last agent that settles has to perform the neighborhood search by itself and determine that it has no empty neighbors left, marking it \textbf{fullyVisited}. Then the DFShead traverses the tree back to the root performing neighborhood search at each node to ensure that it becomes \textbf{fullyVisited}, and reconfigure \textbf{partiallyVisited} neighbors. Note that, the port-one neighbor of a \textbf{partiallyVisited} node is marked \textbf{visited}. It can only be marked \textbf{fullyVisited} when the \textbf{partiallyVisited} node is reconfigured to become a \textbf{visited} node. Hence, on termination, no \textbf{partiallyVisited} nodes remain.

\vspace{2mm}
\noindent\textbf{Time complexity optimization.} 
This straight-forward algorithm takes $2\delta_x$ epochs to perform a neighborhood search at $x$. However, when multiple agents are present at $x$, they can always visit the ports at $x$ in parallel to determine the result corresponding to a port. In Sections~\ref{subsec:vacant} and~\ref{subsec:probe}, we present two methods that go hand-in-hand for performing $O(1)$-epoch neighborhood search at a node $x$. First, some of the nodes are chosen to be \textsc{vacated} such that the settled agents at those nodes could travel with the DFShead to perform the neighborhood search. Second, even in the absence of settled agents at \textsc{vacated} nodes, the parallel probe can correctly determine that it is \textsc{vacated}. 

\subsection{Selecting Vacant Nodes}\label{subsec:vacant}
 Here we describe how vacant nodes are chosen. 
 Once a vacant node $w$ is chosen, the state of $w$ becomes \textsc{vacated}, and the settled agent $\psi(w)$ travels with the DFShead instead of remaining at $w$. 
 However, the agent $\psi(w)$ needs to collect certain information before it can travel with the DFShead. Informally, the condition for designating a \textbf{visited} node $x$ to be \textsc{vacated} is that,
 the port-1 neighbor $w$ must be \textsc{occupied} in {\sc P1Tree} $\mathcal{T}$ constructed so far. The \textbf{fullyVisited} and \textbf{partiallyVisited} nodes are always \textsc{vacated}.

 \vspace{2mm}
 \noindent\textbf{Detailed description.}
 Consider an execution of  Algorithm \ref{alg:portoneDFS} (\texttt{DFS\_P1Tree()}).
 Suppose the DFShead is located at $x$. Let $z$ be the parent of  $x$. On arrival at $x$, the DFShead determines the state of node $x$. Let $w$ be the port-1 neighbor of $x$, i.e, $p_{xw} = 1$. 
 The settled agent $\psi(x)$ stores the information about its port-1 neighbor using the variables: $\psi(x).\textsf{P1Neighbor}  = \psi(w)$ ($\psi(w)$ is $\bot$ when $w$ is \textsc{empty}) and $\psi(x).\textsf{portAtP1Neighbor}  = p_{wx}$. 
 Formally,
 \begin{itemize}[leftmargin=3em]
     \item[(V1)] the root node is always \textsc{occupied}.
     \item[(V2)] $x$ is \textsc{vacated}, if $w$ is \textsc{occupied} and $x$ is \textbf{visited}.
     \item[(V3)] $x$ is \textsc{vacated}, if it is \textbf{fullyVisited} and $\psi(x).\textsf{vacatedNeighbor} = false$.
     \item[(V4)] $x$ is \textsc{vacated}, if it is \textbf{partiallyVisited}.
     \item[(V5)] $z$ is \textsc{vacated}, if $\psi(z).\textsf{vacatedNeighbor} = false$, and  $p_{zx} = 1$.
 \end{itemize}
 When $x$ is \textsc{vacated}, the DFShead moves to $w$ to assign $\psi(w).\textsf{vacatedNeighbor} = true$.
Notice that, $x$ is vacated when $p_{xz} = 1$ (by (V2)), and thus making $\psi(z).\textsf{vacatedNeighbor} = true$. Then, even if $p_{zx} = 1$, rule (V5) is not applicable anymore. This shows the priority order among the rules, and they are applicable in that priority order. The pseudocode is provided in Algorithm~\ref{alg:vacant} \texttt{Can\_Vacate()}, which returns the state \emph{settledScout} for the agent if the node has state \textsc{vacated}.

Fig.~\ref{fig:post} illustrates these ideas. The blue/green nodes are vacant, and gray nodes are occupied. 
Node \textcircled{\scriptsize 0} is \textsc{occupied} since it is the root (by (V1)). Nodes \textcircled{\scriptsize 1}, \textcircled{\scriptsize 3}, and \textcircled{\scriptsize 9} are vacant since their port-1 neighbor \textcircled{\scriptsize 0} is \textsc{occupied} (by (V2)).
Nodes \textcircled{\scriptsize 8}, and \textcircled{\scriptsize 10} are \textsc{vacated} since they are \textbf{fullyVisited} and do not have any dependent \textsc{vacated} neighbors (by (V3)). 
Node \textcircled{\scriptsize 6} is \textsc{vacated} since it is \textbf{partiallyVisited}  (by (V4)).
Node \textcircled{\scriptsize 4} is \textsc{vacated} only after DFShead reaches \textcircled{\scriptsize 5}, and \textcircled{\scriptsize 5} cannot be vacant since port-1 neighbor of \textcircled{\scriptsize 5} (node \textcircled{\scriptsize 1}) is \textsc{vacated} already (by (V5)).

\begin{lemma}\label{lem:threeconsecutive}
    Consider three consecutive nodes $v_1, v_2,$ and $v_3$ visited by the DFShead. Suppose $v_1$ was visited for the first time. Then at least one of $v_1, v_2,$ and $v_3$ is \textsc{vacated}.
\end{lemma}
\begin{proof}
    We prove this by contradiction. For the sake of contradiction, assume that $v_1$, $v_2$, and $v_3$ are \textsc{occupied}. The edges traversed by the DFShead are $\{v_1, v_2\}$ and $\{v_2, v_3\}$. 
    We have the following cases.
    \begin{itemize}
        \item If $v_2$ is the parent of $v_1$, then DFShead is backtracking from $v_1$, which implies, $v_1$ is either \textbf{fullyVisited} or \textbf{partiallyVisited}. Then $v_1$ is \textsc{vacated} by rules (V3) or (V4).
        \item If $v_1 = v_3$, then when the DFShead is at $v_2$, it backtracks, which implies that $v_2$ is either \textbf{fullyVisited} or \textbf{partiallyVisited}. Then $v_2$ is \textsc{vacated} by rules (V3) or (V4).
        \item Otherwise, there is a parent-child path $v_1 \to v_2 \to v_3$, such that $v_2$ is a child of $v_1$. Now, $v_2$ must be \textbf{visited}. Let $u_2$ be the port-1 neighbor of $v_2$. 
        \begin{itemize}
            \item If $u_2 = v_1$, then $v_2$ is \textsc{vacated} by rule (V2) when DFShead reaches $v_2$.
            \item If $u_2 = v_3$, then $u_2$ is \textsc{empty} when DFShead reaches $v_2$ and \texttt{type}$(\{v_2,v_3\})$ is either \texttt{t11} or \texttt{t1q}. Since $v_2$ was \textsc{empty} before the first visit of DFShead, thus $\psi(v_2).\textsf{vacatedNeighbor}$ must be $false$. Thus, $v_2$ is \textsc{vacated} by rule (V5) once DFShead reaches $v_3$ if  \texttt{type}$(\{v_2,v_3\}) =$ \texttt{t1q} or $v_3$ is \textsc{vacated} by rule (V2) if \texttt{type}$(\{v_2,v_3\}) =$ \texttt{t11}. 
            \item If $u_2 (\neq v_3)$ is \textsc{empty}, and since the DFShead moves according to the edge priority, \texttt{type}$(\{v_2,v_3\})$ must be \texttt{tp1}, i.e., $p_{v_3v_2} =1$. Then $v_3$ would be \textsc{vacated} by rule (V2).
            \item If $v_1$ is the root, then $v_1$ remains \textsc{occupied} by rule (V1), however, $u_2$ must be \textsc{empty} since $u_2 \neq v_1$. This falls in the case of $u_2$ is \textsc{empty}.
            \item Consider $u_2\neq v_1$ and $u_2\neq v_3$. If $u_2$ is \textsc{vacated}, we have the following cases.
            \begin{itemize}
                \item If \texttt{type}$(\{v_1,v_2\})$ is \texttt{t1q}, 
                then once DFShead reaches $v_2$, rule (V5) is applicable to $v_1$ and is \textsc{vacated}. This is possible since $v_1$ was \textsc{empty} before the first visit of DFShead and thus $\psi(v_1).\textsf{vacatedNeighbor}$ is $false$.
                \item If \texttt{type}$(\{v_1,v_2\})$ is \texttt{tpq}, then \texttt{type}$(\{v_2,v_3\})$ cannot be \texttt{tpq}, because then \texttt{type}$(v_2)$ would be \textbf{partiallyVisited}. When \texttt{type}$(\{v_2,v_3\})$ is either \texttt{t11} or \texttt{t1q}, it is the case $u_2 = v_3$. If \texttt{type}$(\{v_2,v_3\})$ is \texttt{tp1}, then $v_3$ is \textsc{vacated} once DFShead reaches $v_3$ by rule (V2).
            \end{itemize}
        \end{itemize}
    \end{itemize}
    In each of the cases, we obtain a contradiction by showing at least one of $v_1$, $v_2$ or $v_3$ is \textsc{vacated}. Hence proved.
\end{proof}

\begin{lemma}\label{lem:vacated}
Let $\mathcal{T}$ be the partial tree constructed by Algorithm~\ref{alg:portoneDFS} (\texttt{DFS\_P1Tree()}) at a given moment, and let $k = |\mathcal{T}|$ be the number of vertices in $\mathcal{T}$. Then at least $\lfloor k/3\rfloor$ of these $k$ vertices are in state \textsc{vacated}.
\end{lemma}

\begin{proof}
We define $v_i \prec v_j$, if $v_i$ is visited before $v_j$ by the DFShead for the first time.
Consider the vertices in this order as $D_v = (v_1,v_2,\dots,v_k)$, where $v_i \prec v_{i+1}$ for $1 \leq i <k$.  
Let $v_i'$ and $v_i''$ be the two subsequent nodes visited by DFShead immediately after $v_i$.
Now, we define $B(v_i)$ be the corresponding block of $v_i$ comprising of nodes in $(v_i,v_i',v_i'')$ if they appear after $v_i$ in the DFS order.
Specifically, if $v_i' \prec v_i$, then $v_i'$ does not belong to $B(v_i)$; likewise for $v_i''$. 
We say $B(v_i)$ is degenerate if $|B(v_i)| < 3$. By default, $|B(v_i)| \ge 1$ for all $1 \le i \le k$, since $v_i \in B(v_i)$.
If $B(v_i)$ is not degenerate, then it contains $v_i, v_{i+1}$, and  $v_{i+2}$.

We determine a subsequence $S_v = (\bar{v}_0,\dots, \bar{v}_l)$ of $D_v$ iteratively, 
(i) $v_0 = \bar{v}_0$; (ii) $\bar{v}_j = v_i$ where $v_i$ is the first element in $D_v \setminus \bigcup_{\alpha= 0}^{j-1}  B(\bar{v}_\alpha)$.
We have $l \geq \lceil k/3\rceil$, since $B(v_i) \leq 3$.
When $B(\bar{v}_j) < 3$ for $j < l$, then the DFShead backtracks from $\bar{v}_j$, i.e., $\bar{v}_j$ is \textsc{vacated}. Also, when $B(\bar{v}_j) = 3$, at least one node in $B(\bar{v}_j)$ is \textsc{vacated} by Lemma~\ref{lem:threeconsecutive}.

Finally, consider $B(\bar{v}_l)$. If $B(\bar{v}_l) = 3$, then we have a \textsc{vacated} node in $B(\bar{v}_l) = 3$. 
If $B(\bar{v}_l) < 3$, then there may not be \textsc{vacated} nodes in $B(\bar{v}_l)$.
Overall, we have at least $l-1$ \textsc{vacated} nodes, where $l - 1 = \lfloor k/3\rfloor$. Hence proved.
\end{proof}

%% file: terminology.tex
\vspace{2mm}
\noindent{\bf Some terminologies.}
Throughout the paper,  we encode an edge $\{u, v\}$  by the 4-tuple
\[
e \;=\;[\,u,\,p_{uv},\,p_{vu},\,v\,],\qquad 
1\le p_{uv}\le\delta_{u},\;1\le p_{vu}\le\delta_{v},
\]
where $u,v$ are the (anonymous) end-nodes and
$p_{uv}$ (resp.\ $p_{vu}$) is the port number of $e$ at $u$ (resp.\ $v$).
We define type of edge $\{u,v\}$ at $u$ as follows:
\begin{itemize}
    \item \texttt{t11}, if $p_{uv}=1=p_{vu}$
    \item \texttt{tp1}, if $p_{uv}\neq 1=p_{vu}$
    \item \texttt{t1q}, if $p_{uv}=1\neq p_{vu}$
    \item \texttt{tpq}, if $p_{uv}\neq 1\neq p_{vu}$.
\end{itemize}
We denote the type of an edge $\{u,v\}$ by $\texttt{type}(\{u,v\}) = \texttt{type}(p_{uv},p_{vu}) 
\in \{\texttt{tp1},\texttt{t11},\texttt{t1q},\texttt{tpq}\}$.

%% file: fig_examplepost.tex
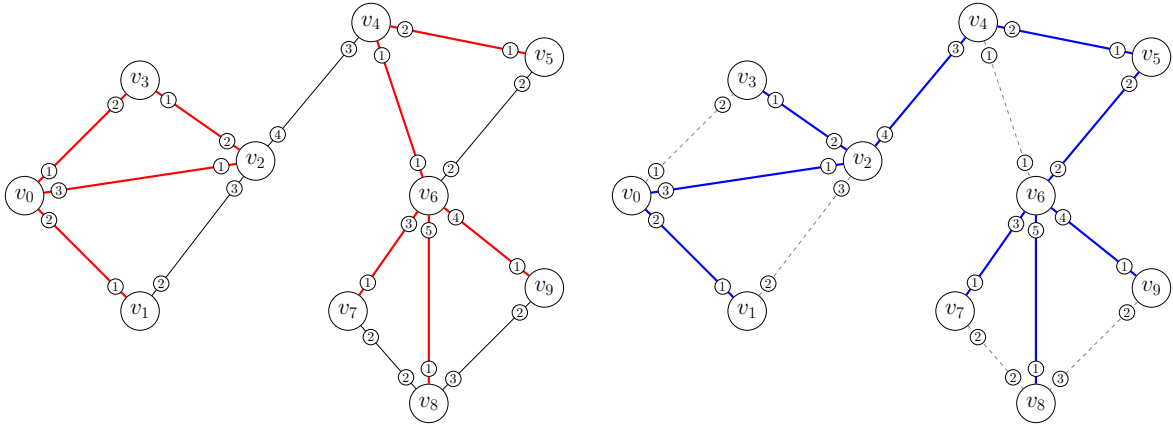
\begin{figure}[htbp] 
    \centering
\ifshowtikz

    \resizebox{0.45\textwidth}{!}{%
    \begin{tikzpicture}[
        scale=3,
        every node/.style={circle, draw=black, fill=white, minimum size=10mm, text=black, inner sep=1pt, font=\LARGE},
        port/.style={circle, fill=white, draw=black, text=black, font=\small, inner sep=1pt, minimum size=4mm},
        highlight_port1_style/.style={draw=red, ultra thick}, 
        default_edge_style/.style={draw=black}   
      ]

      \node (n0)  at (0,0)    {$v_0$};
      \node (n1)  at (1,1)    {$v_3$};
      \node (n2)  at (1,-1)   {$v_1$};
      \node (n3)  at (2, 0.3) {$v_2$};
      \node (n4)  at (3, 1.5) {$v_4$};
      \node (n5)  at (3.5, 0) {$v_6$};
      \node (n6)  at (4.5, 1.2) {$v_5$};
      \node (n7)  at (2.8, -1) {$v_7$};
      \node (n8)  at (3.5, -1.8){$v_8$};
      \node (n9)  at (4.5, -0.8){$v_9$};

      \foreach \source/\dest/\psrc/\pdst/\elabel/\isTree in {
          n0/n1/1/2/e1/1,   
          n0/n2/2/1/e2/1,   
          n0/n3/3/1/e3/1,   
          n1/n3/1/2/e4/0, 
          n2/n3/2/3/e5/0, 
          n3/n4/4/3/e6/1,   
          n4/n6/2/1/e7/0,   
          n4/n5/1/1/e8/1,   
          n5/n6/2/2/e9/1, 
          n5/n7/3/1/e10/1,  
          n5/n9/4/1/e11/1,  
          n5/n8/5/1/e12/0,  
          n7/n8/2/2/e13/1,
          n9/n8/2/3/e14/0 
      } {
          \def\currentedgestyle{default_edge_style} 
          \ifnum \psrc = 1
              \def\currentedgestyle{highlight_port1_style} 
          \else
              \ifnum \pdst = 1
                  \def\currentedgestyle{highlight_port1_style} 
              \fi
          \fi

          \draw[\currentedgestyle] (\source) -- (\dest);

          \node[port] at ($(\source)!3mm!(\dest)$) {\psrc};
          \node[port] at ($(\dest)!3mm!(\source)$) {\pdst};
      }

    \end{tikzpicture}%
    }
    \hspace{1em} 
    \resizebox{0.45\textwidth}{!}{%
    \begin{tikzpicture}[
        scale=3,
        every node/.style={circle, draw=black, fill=white, minimum size=10mm, text=black, inner sep=1pt, font=\LARGE},
        port/.style={circle, fill=white, draw=black, text=black, font=\small, inner sep=1pt, minimum size=4mm},
        tree_edge_style/.style={draw=blue, ultra thick}, 
        non_tree_edge_style/.style={draw=gray, dashed, thin} 
      ]

      \node (n0)  at (0,0)    {$v_0$};
      \node (n1)  at (1,1)    {$v_3$};
      \node (n2)  at (1,-1)   {$v_1$};
      \node (n3)  at (2, 0.3) {$v_2$};
      \node (n4)  at (3, 1.5) {$v_4$};
      \node (n5)  at (3.5, 0) {$v_6$};
      \node (n6)  at (4.5, 1.2) {$v_5$};
      \node (n7)  at (2.8, -1) {$v_7$};
      \node (n8)  at (3.5, -1.8){$v_8$};
      \node (n9)  at (4.5, -0.8){$v_9$};

       \foreach \source/\dest/\psrc/\pdst/\elabel/\isTree in {
          n0/n1/1/2/e1/0,
          n0/n2/2/1/e2/1,
          n0/n3/3/1/e3/1,
          n1/n3/1/2/e4/1,
          n2/n3/2/3/e5/0,
          n3/n4/4/3/e6/1,
          n4/n6/2/1/e7/1,
          n4/n5/1/1/e8/0,
          n5/n6/2/2/e9/1,
          n5/n7/3/1/e10/1,
          n5/n9/4/1/e11/1,
          n5/n8/5/1/e12/1,
          n7/n8/2/2/e13/0,
          n9/n8/2/3/e14/0
      } {
          \ifnum \isTree = 1 
              \draw[tree_edge_style] (\source) -- (\dest);
              \node[port] at ($(\source)!3mm!(\dest)$) {\psrc};
              \node[port] at ($(\dest)!3mm!(\source)$) {\pdst};
          \else 
              \draw[non_tree_edge_style] (\source) -- (\dest);
              \node[port] at ($(\source)!3mm!(\dest)$) {\psrc};
              \node[port] at ($(\dest)!3mm!(\source)$) {\pdst};
          \fi 
      }

    \end{tikzpicture}%
    }
\fi
    \caption{An example of a Port-One Tree. Left: Edges incident to a port 1 are highlighted in red. Right: Tree edges (blue, solid) shown with non-tree edges (gray, dashed).}
    \label{fig:examplepost}
\end{figure}

%% file: fig_post.tex
\begin{figure}
    \centering
    \begin{subfigure}[t]{0.48\textwidth}
        \centering
        \resizebox{\textwidth}{!}{%
\begin{tikzpicture}[
    every node/.style={circle, draw=black, fill=gray!70, minimum size=10mm, text=white, font=\LARGE},
    port/.style={circle, fill=white, draw=black, text=black, font=\footnotesize, inner sep=1pt, minimum size=12pt},
    edge/.style={draw=gray!70, thick, dashed}, 
    tree_edge/.style={draw=gray!70, very thick} 
]
\node[fill=red!70] (n0) at (0,0) {0};
\node (n1) at (-10,-1) {2};
\node[fill=blue!70] (n2) at (3,-3) {4};
\node[fill=blue!70] (n3) at (-5,-3) {9};
\node[fill=blue!70] (n4) at (-6,3) {1};
\node (n5) at (2,4) {5};
\node[fill=blue!70] (n6) at (-9,-5) {3};
\node[fill=blue!70] (n7) at (6,0) {7};
\node[fill=cyan!70] (n8) at (0,-5) {10};
\node[fill=cyan!70] (n9) at (4,2) {8};
\node[fill=green,text=black] (n10) at (-2,6) {6};

\foreach \src/\dst/\pu/\qv/\tree/\pone in {
    n0/n3/1/1/0/0,
    n0/n4/2/1/1/1,
    n0/n7/4/1/0/0,
    n0/n8/5/2/0/0,
    n0/n10/6/1/0/0,
    n1/n3/2/2/1/0,
    n1/n4/1/2/1/1,
    n1/n6/3/1/1/1,
    n2/n5/1/2/1/1,
    n2/n6/2/2/1/0,
    n2/n7/3/2/1/0,
    n3/n8/4/1/1/1,
    n3/n10/5/2/0/0,
    n4/n5/3/1/0/0,
    n4/n10/4/4/0/0,
    n5/n10/3/3/1/0,
    n7/n9/3/1/1/1
} {
    \ifnum\tree=1
        \ifnum\pone=1
            \draw[tree_edge, blue] (\src) -- (\dst);
        \else
            \draw[tree_edge, red] (\src) -- (\dst);
        \fi
    \else
        \draw[edge] (\src) -- (\dst);
    \fi
    
    \node[port] at ($(\src)!0.8cm!(\dst)$) {\pu};
    \node[port] at ($(\dst)!0.8cm!(\src)$) {\qv};
}
\end{tikzpicture}}
        \caption{DFS backtracks at \textcircled{\scriptsize 6} marking it \textbf{partiallyVisited} since parent edge $\{\textcircled{\scriptsize 6},\textcircled{\scriptsize 5}\}$ is of type \texttt{tpq} and its only \textbf{unvisited} neighbor $\textcircled{\scriptsize 9}$ is connected via a \texttt{tpq} edge $\{\textcircled{\scriptsize 6},\textcircled{\scriptsize 9}\}$.}
        \label{fig:post}
    \end{subfigure}
    \hfill
    \begin{subfigure}[t]{0.48\textwidth}
        \centering
        \resizebox{\textwidth}{!}{%
\begin{tikzpicture}[
    every node/.style={circle, draw=black, fill=gray!70, minimum size=10mm, text=white, font=\LARGE},
    port/.style={circle, fill=white, draw=black, text=black, font=\footnotesize, inner sep=1pt, minimum size=12pt},
    edge/.style={draw=gray!70, thick, dashed}, 
    tree_edge/.style={draw=gray!70, very thick} 
]
\node[fill=red!70] (n0) at (0,0) {0};
\node (n1) at (-10,-1) {2};
\node[fill=blue!70] (n2) at (3,-3) {4};
\node[fill=blue!70] (n3) at (-5,-3) {9};
\node[fill=blue!70] (n4) at (-6,3) {1};
\node (n5) at (2,4) {5};
\node[fill=blue!70] (n6) at (-9,-5) {3};
\node[fill=blue!70] (n7) at (6,0) {7};
\node[fill=cyan!70] (n8) at (0,-5) {10};
\node[fill=cyan!70] (n9) at (4,2) {8};
\node[fill=blue!70] (n10) at (-2,6) {6};

\foreach \src/\dst/\pu/\qv/\tree/\pone in {
    n0/n3/1/1/0/0,
    n0/n4/2/1/1/1,
    n0/n7/4/1/0/0,
    n0/n8/5/2/0/0,
    n0/n10/6/1/1/1,
    n1/n3/2/2/1/0,
    n1/n4/1/2/1/1,
    n1/n6/3/1/1/1,
    n2/n5/1/2/1/1,
    n2/n6/2/2/1/0,
    n2/n7/3/2/1/0,
    n3/n8/4/1/1/1,
    n3/n10/5/2/0/0,
    n4/n5/3/1/0/0,
    n4/n10/4/4/0/0,
    n5/n10/3/3/0/0,
    n7/n9/3/1/1/1
} {
    \ifnum\tree=1
        \ifnum\pone=1
            \draw[tree_edge, blue] (\src) -- (\dst);
        \else
            \draw[tree_edge, red] (\src) -- (\dst);
        \fi
    \else
        \draw[edge] (\src) -- (\dst);
    \fi
    
    \node[port] at ($(\src)!0.8cm!(\dst)$) {\pu};
    \node[port] at ($(\dst)!0.8cm!(\src)$) {\qv};
}
\end{tikzpicture}}
        \caption{Reconfiguration of the P1Tree $\mathcal{T}$ for a \textbf{partiallyVisited} node removing the $\texttt{tpq}$ edge $\{\textcircled{\scriptsize 6},\textcircled{\scriptsize 5}\}$ and adding the edge  $\{\textcircled{\scriptsize 0},\textcircled{\scriptsize 6}\}$ of type \texttt{tp1}.}
        \label{fig:reconfiguredpost}
    \end{subfigure}
    \caption{An illustration of reconfiguration on a {\sc P1Tree} $\mathcal{T}$ constructed so far, swapping  a $\texttt{tpq}$ edge by an edge of type \texttt{tp1} or \texttt{t11}.}
    \label{fig:reconfig}
\end{figure}
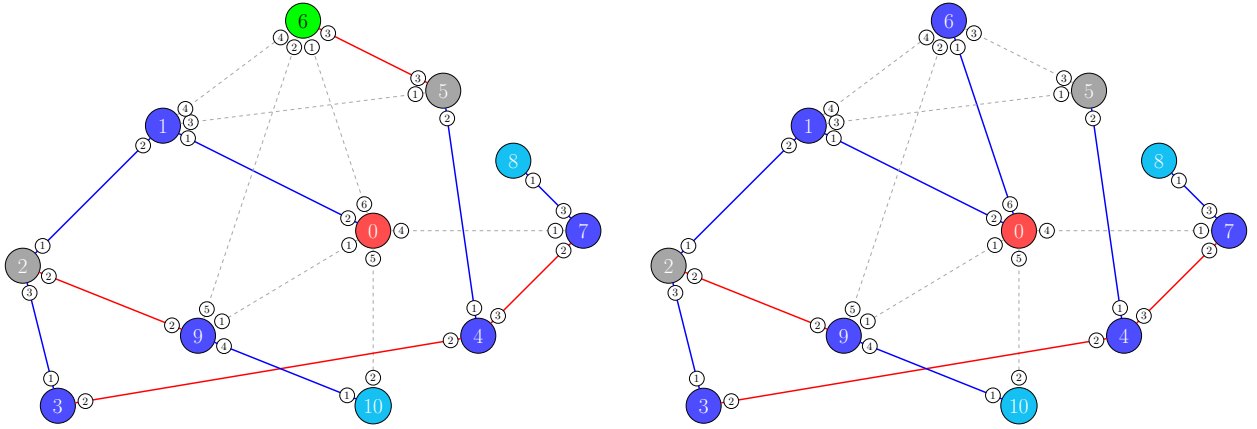

%% file: probe.tex
\subsection{Parallel Probing}
\label{subsec:probe}
Now, we present the core technique on which our {\dis} algorithm hinges. 
We have selected the nodes that are designated as \textsc{vacated}. The agents settled at those nodes travel with the DFShead to help with the neighborhood search. 
Here, we show to use the agents at DFShead to perform a neighborhood search, called {\em Parallel Probing}, to determine the state of neighbors from \textsc{empty}, \textsc{vacated} and \textsc{occupied}; and obtain the scout result corresponding to a port $p_{xy}$ as the 4-tuple $\langle p_{xy}, \texttt{type}(\{x,y\}), \texttt{type}(y), \psi(y) \rangle$. The pseudocode of the algorithm is given in Algorithm \ref{proc:probe} which we call \texttt{Parallel\_Probe()}.

\vspace{2mm}
\noindent{\bf Highlevel overview.}
The core idea of parallel probing stems from the fundamental observation that each node has a port-1 neighbor. As we have seen in Section~\ref{subsec:vacant}, a \textbf{visited} node is \textsc{vacated} when its port-1 neighbor is occupied or when a node is \textbf{partiallyVisited} or \textbf{fullyVisited}; we call these agents settled scouts. In \textsc{occupied} nodes the settled agent remains, but in \textsc{vacated} nodes the settled agent travels with the DFShead.

Consider a node $x$ where the DFShead is doing a neighborhood search. DFShead first settles an agent at $x$. 
The unsettled agents and settled scouts perform the neighborhood search.
The ports at $x$ are assigned to agents in the increasing order of their IDs until all ports are probed.
An agent visiting a neighbor $y$ (i.e., probing port $p_{xy}$), determines whether $y$ is \textsc{empty}, \textsc{occupied} or \textsc{vacated}.
By default, if $y$ is \textsc{occupied}, then the probing agent returns directly.
The challenge arises when the probing agent must distinguish between a node that is \textsc{empty} vs. \textsc{vacated}.
When a \textbf{visited} node is \textsc{vacated}, its port-1 neighbor is \textsc{occupied} in the tree.
Exploiting this fact, the probing agent visits the port-1 neighbor $z$ of $y$. If $z$ is occupied, then it returns to $x$. At $x$, it checks among the other agents present, if there exists an agent $b$ such that port-1 neighbor of $b$ is $\psi(z)$. 
However, the \textsc{vacated} node $y$ can be \textbf{partiallyVisited} or \textbf{fullyVisited}.
In that case, $z$ may be a \textbf{visited} node that is \textsc{vacated}. Then the probing agent visits port-1 neighbor $w$ of $z$, and returns to $x$ if $\psi(w)$ is present. 
The probing agent can then transitively check for an agent $c$ such that $\psi(w)$ is $c$'s port-1 neighbor; and $b$ such that $c$ is $b$'s port-1 neighbor.

\vspace{2mm}
\noindent{\bf Detailed description.}
Let $x$ be the DFShead with settled agent $\psi(x)$. Let its parent port be \(\psi(x).\textsf{parentPort}\) ($\bot$
at the root).
For every
\(
p_{xy}\in\{1,\dots,\delta_x\}\setminus\{\psi(x).\textsf{parentPort}\}
\)
the head chooses a scout agent \(a\in A_{scout}\) in the increasing order of their ID to scout the ports leading to $N(v)$. 
To describe the rules in a simple manner, we consider the neighbor $y$ to be a variable.
Scout $a$ learns the edge type \(\texttt{type}(\{x,y\}) \in \{\texttt{tpq, tp1, t11, t1q}\}\) when it reaches the neighbor $y\in N(v)$, based on the arrival port $p_{yx}$.
Upon arrival the scout sets
$a.\textsf{scoutEdgeType}\;\gets\;\texttt{type}(\{x,y\})$. 
The agent $a$ can also determine $\xi(y)$ as the settled node present at $y$. 
The node $y$ can be either \textsc{empty}, \textsc{vacated} or \textsc{occupied}.
Note that, when $\psi(y)$ is \textsc{vacated} or \textsc{empty}, $\xi(y)$ is $\bot$.
The node type of $y$ is stored at $\psi(y).\textsf{nodeType}$ (invalid when $y$ is \textbf{unvisited}, and assigned $\bot$). 
At node $x$, the probe rules are as follows to determine $\psi(y)$. Once $\psi(y)$ is determined, the scout result is stored as $a.\textsf{scoutResult} \gets \langle p_{xy}, a.\textsf{scoutEdgeType}, \psi(y).\textsf{nodeType}, \psi(y) \rangle$.
\begin{Ren}\sloppy
\item If $\xi(y)\neq\bot$, then $\psi(y)=\xi(y)$ and return to $x$.

\item If $\xi(y)=\bot$ and $p_{yx}=1$, then $\psi(y)=\bot$ and return to $x$.

\item If $\xi(y)=\bot$ and $p_{yx}\neq1$, let $z$ be the port-1 neighbor of $y$. Go to $z$. Store $a.\textsf{scoutP1Neighbor}=\xi(z)$ and
            $a.\textsf{scoutPortAtP1Neighbor}=p_{zy}$.
      \begin{Ren}
      \item If $\xi(z)\neq\bot$, return to $x$ and check whether
            $\exists\,b\in A_{scout}$ with
            $b.\textsf{P1Neighbor}=\xi(z)$ and
            $b.\textsf{portAtP1Neighbor}=p_{zy}$.
            \begin{Ren}
            \item If such $b$ exists, then $\psi(y)=b$.
            \item Otherwise $\psi(y)=\bot$.
            \end{Ren}

      \item
            If $\xi(z)=\bot$ and $p_{zy}=1$, then $\psi(y)=\bot$ and return to $x$.

      \item
            If $\xi(z)=\bot$ and $p_{zy}\neq1$, visit the port-1 neighbor $w$ of $z$.
            Store $a.\textsf{scoutP1P1Neighbor}=\xi(w)$ and
            $a.\textsf{scoutPortAtP1P1Neighbor}=p_{wz}$.
            \begin{Ren}
            \item
                  If $\xi(w)=\bot$, then $\psi(y)=\bot$.
            \item
                  If $\xi(w)\neq\bot$, return to $x$ and check  
                  (i) $\exists\,c\in A_{scout}$ with
                  $c.\textsf{port1Neighbor}=\xi(w)$ and
                  $c.\textsf{portAtP1Neighbor}=p_{wz}$, and  
                  (ii) $\exists\,b\in A_{scout}$ with
                  $b.\textsf{P1Neighbor}=c$ and
                  $b.\textsf{portAtP1Neighbor}=p_{zy}$.
                  \begin{itemize}
                  \item[$(\alpha)$] If both $c$ and $b$ exist, then $\psi(y)=b$.
                  \item[$(\beta)$] Otherwise $\psi(y)=\bot$.
                  \end{itemize}
            \end{Ren}
      \end{Ren}
\end{Ren}

Note that, the checking for presence of another agent that was originally \textsc{vacated} happens only after all the agents in $A_{scout}$ return to $x$.

Fig.~\ref{fig:probe} shows one instance for every rule from
\textbf{(R1)} to \textbf{(R3c-ii)}.  
The central red node $x$ is the position of DFShead performing the parallel probe.
Small white circles carry local port numbers.
\input{fig_probe}

\begin{lemma}
\label{lemma:probeone}
  Algorithm \ref{proc:probe} (\texttt{Parallel\_Probe()}) at a node $x\in V$ correctly determines the state of a neighbor node in $O(1)$ epochs.
\end{lemma}
\begin{proof}
    On running Algorithm \ref{proc:probe} (\texttt{Parallel\_Probe()}) at $x$, the state of a neighbor $y$ is clearly determined when $\xi(y) \neq \bot$. The settled agent at $y$ remains at $y$, and hence it results in \textsc{occupied}.
    The main challenge of correctly determining state is identifying the \textsc{vacated} neighbors, since by default the result assumes $y$ to be \textsc{empty}.
    When $y$ is \textsc{vacated}, it can be either \textbf{visited}, \textbf{fullyVisited} or \textbf{partiallyVisited}. We handle each of the cases separately.
    \begin{itemize}
        \item When $y$ is \textbf{visited} and \textsc{vacated}, then port-1 neighbor of $y$ must be \textsc{occupied} (rule (V2)). Hence by visiting the port-1 neighbor $z$ of $y$, the agent will find $\xi(z) \neq \bot$. When $y$ was \textsc{vacated}, $\psi(y)$ must have stored $\xi(z)$ as $\psi(y).\textsf{P1Neighbor}$. Thus, after the agent returns to $x$, the agent $\psi(y)$ must be at $x$ since it is travelling with the DFShead.
        \item When $y$ is \textbf{partiallyVisited} and \textsc{vacated}, that means, the port-1 neighbor $z$ of $y$ can either be \textsc{occupied} or \textsc{vacated}. When $\xi(z)\neq \bot$, we can determine the presence $\psi(y)$ in $A_{scout}$ by checking if $\xi(z)$ is the port-1 neighbor of $\psi(y)$. Furthermore, when $z$ is also \textsc{vacated}, then $z$ must be of type \textbf{visited}. Since $y$ can be \textbf{partiallyVisited} if only if there are no empty neighbors of $y$ via \texttt{tp1}, \texttt{t1q} or \texttt{t11} type edges when DFShead first reached $y$. Since $z$ is connected to $y$ via either \texttt{tp1} or \texttt{t11} edge, $z$ must be of type \textbf{visited} since it had an empty neighbor ($y$) when DFShead was at $z$.
        Thus the port-1 neighbor $w$ of $z$ must be \textsc{occupied} when $z$ is \textsc{vacated}. Thus, when $y$ is \textbf{partiallyVisited}, either $z$ or $w$ must be \textsc{occupied}. The agent after returning to $x$ can check for existence of $\psi(z)$ and $\psi(y)$ in $A_{scout}$ and determine $\psi(y)$.
        \item When $y$ is \textbf{fullyVisited}, it is chosen to be \textsc{vacated} only if $\psi(y).\textsf{vacatedNeighbor}$ is $false$ (by (V3)). When DFShead is at $y$, it must have no \textsc{empty} neighbors. The port-1 neighbor $z$ of $y$ must be \textbf{visited} since when DFShead was at $z$, it had an edge of type \texttt{tp1} or \texttt{t11} leading to $y$ (which was \textsc{empty}). Analogous to the previous case, since $z$ is a \textbf{visited} node, if it is \textsc{vacated}, then the port-1 neighbor of $z$ must be \textsc{occupied}. Hence the agent can transitively determine $\psi(y)$ at $x$ from the agents in $A_{scout}$.
    \end{itemize}
    Certainly, the agent from a \textsc{vacated} node maybe probing another port. We can say that all probing agents spend at most 6 epochs before returning to $x$. Hence in $O(1)$ epochs, all probing agents are at $x$. At that time, all agents can determine the existence of other agents in $A_{scout}$.
\end{proof}

\begin{lemma}
\label{lemma:parallelprobe}
  Algorithm \ref{proc:probe} (\texttt{Parallel\_Probe()}) at a node $x\in V$ determines the state of $O(|A_{scout}|)$ neighbor nodes in $O(1)$ epochs.
\end{lemma}
\begin{proof}
    As we observe in Lemma~\ref{lemma:probeone}, it takes $O(1)$ edge traversals (thus epochs) to determine the state of a neighbor. Each agent in $A_{scout}$ can probe in parallel to determine the state of $|A_{scout}|$ neighbors in $O(1)$ epochs. Hence $O(|A_{scout}|)$ neighbor states can be identified in $O(1)$ epochs.
\end{proof}

%% file: fig_probe.tex
\begin{figure}[!t]
\centering
\ifshowtikz
\resizebox{\textwidth}{!}{
\begin{tikzpicture}[
  node/.style  = {circle,draw=black,minimum size=7.5mm,inner sep=0},
  port/.style  = {circle,draw=black,fill=white,font=\tiny,inner sep=0pt,minimum size=8pt},
  edge/.style  = {gray!70,thick},
  x node/.style={node,fill=red!60},  occ/.style={node,fill=gray!35},
  vis/.style={node,fill=blue!40},   part/.style={node,fill=green!40},
  full/.style={node,fill=cyan!40}, emp/.style={node,fill=white},
  lbl/.style={font=\small}]
\def\R{3.5cm}

\node[x node] (x) at (0,0) {$x$};

\node[occ]   (y1)    at  ( 90:\R)            {$y_{1}$};
\node[emp]   (y2)    at  ( -90:\R)            {$y_{2}$};
\node[vis]   (y3)    at  ( 30:\R)            {$y_{3a}$};
\node[part]   (y3a1)  at  (  0:\R)            {$y_{3a}'$};
\node[full]  (y3a2)  at  (-30:\R)            {$y_{3a}''$};
\node[emp]  (y3a3)  at  (-60:\R)            {$y_{3a}'''$};
\node[emp]   (y3b)   at  ($(60:\R)$)            {$y_{3b}$};
\node[emp]   (y3c1a) at  (-120:\R)           {$y_{3c}'$};
\node[emp]   (y3c2a) at  (-150:\R)           {$\bar{y}_{3c}$};
\node[emp]  (y3c2b) at  ( 180:\R)           {$\bar{y}_{3c}'$};
\node[part]   (y3c2c) at  ( 150:\R)           {$\bar{y}_{3c}''$};
\node[full]   (y3c2d) at  ( 120:\R)          {$\bar{y}_{3c}'''$};

\node[occ] (z3)      at ($(y3)+(3,0)$)         {$z_{3a}$};
\node[occ] (z3a1)    at ($(y3a1)+(3,0)$)       {$z_{3a}'$};
\node[occ] (z3a2)    at ($(y3a2)+(3,0)$)       {$z_{3a}''$};
\node[occ] (z3a3)    at ($(y3a3)+(3,0)$)       {$z_{3a}'''$};
\node[emp] (z3b)     at ($(y3b)+(3,0)$)        {$z_{3b}$};
\node[emp] (z3c1a)   at ($(y3c1a)+(-3,0)$)     {$z_{3c}$};
\node[emp] (z3c2a)   at ($(y3c2a)+(-3,0)$)     {$\bar{z}_{3c}$};
\node[vis] (z3c2b)   at ($(y3c2b)+(-3,0)$)     {$\bar{z}_{3c}'$};
\node[vis] (z3c2c)   at ($(y3c2c)+(-3,0)$)     {$\bar{z}_{3c}''$};
\node[vis] (z3c2d)   at ($(y3c2d)+(-3,0)$)     {$\bar{z}_{3c}'''$};

\node[emp] (w3c1a)   at ($(z3c1a)+(-3,0)$)     {$w_{3c}'$};
\node[occ] (w3c2a)   at ($(z3c2a)+(-3,0)$)     {$\bar{w}_{3c}$};
\node[occ] (w3c2b)   at ($(z3c2b)+(-3,0)$)     {$\bar{w}_{3c}'$};
\node[occ] (w3c2c)   at ($(z3c2c)+(-3,0)$)     {$\bar{w}_{3c}''$};
\node[occ] (w3c2d)   at ($(z3c2d)+(-3,0)$)     {$\bar{w}_{3c}'''$};

\foreach \i/\p/\q in {1/1/1,2/7/1,3/3/2,3a1/4/2,3a2/5/2,3a3/6/2,3b/2/2,3c1a/8/2,%
                      3c2a/9/2,3c2b/10/2,3c2c/11/2, 3c2d/12/2}
  {
    \draw[edge] (x)--(y\i);
    \node[port] at ($(x)!0.8cm!(y\i)$) {\p};
    \node[port] at ($(y\i)!0.6cm!(x)$) {\q};
  }

\foreach \i/\p/\q in {3/1/2,3a1/1/2,3a2/1/1,3a3/1/1,3b/1/3,3c1a/1/2,
                      3c2a/1/2,3c2b/1/2,3c2c/1/2,3c2d/1/2}
  {
    \draw[edge] (y\i)--(z\i);
    \node[port] at ($(y\i)!0.6cm!(z\i)$) {\p};
    \node[port] at ($(z\i)!0.6cm!(y\i)$) {\q};
  }

\foreach \j in {3c1a,3c2a,3c2b,3c2c,3c2d}
  {
    \draw[edge] (z\j)--(w\j);
    \node[port] at ($(z\j)!0.6cm!(w\j)$) {1};
    \node[port] at ($(w\j)!0.6cm!(z\j)$) {1};
  }

\node[lbl] at ($(y1)+(0,0.55)$)                {\textbf{R1}};
\node[lbl] at ($(y2)+(0,-0.55)$)                {\textbf{R2}};
\node[lbl] at ($(y3)+(0.8,0.4)$)               {\textbf{R3a-i}};
\node[lbl] at ($(y3a1)+(0.8,0.4)$)             {\textbf{R3a-i}};
\node[lbl] at ($(y3a2)+(0.8,0.4)$)             {\textbf{R3a-i}};
\node[lbl] at ($(y3a3)+(0.8,0.4)$)             {\textbf{R3a-ii}};
\node[lbl] at ($(y3b)+(0,0.55)$)              {\textbf{R3b}};
\node[lbl] at ($(y3c1a)+(0,-0.55)$)            {\textbf{R3c-i}};
\node[lbl] at ($(y3c2a)+(0,-0.55)$)            {\textbf{R3c-ii}$(\beta)$};
\node[lbl] at ($(y3c2b)+(-1,-0.35)$)           {\textbf{R3c-ii}$(\beta)$};
\node[lbl] at ($(y3c2c)+(0,0.55)$)             {\textbf{R3c-ii}$(\alpha)$};
\node[lbl] at ($(y3c2d)+(0,0.55)$)             {\textbf{R3c-ii}$(\alpha)$};

\node[vis,node,scale=0.6] (leg3) at ([yshift=1.5cm,xshift=-2.5cm]y1.north) {};
\node[anchor=west] (leg3txt) at (leg3.east) {\textbf{visited} \& \textsc{vacated}};

\node[emp,node,scale=0.6,anchor=east] (leg1) at ([xshift=-6cm, yshift=0.2cm]leg3.west) {};
\node[anchor=west] (leg1txt) at (leg1.east) {\textbf{unvisited} \&  \textsc{empty}};

\node[occ,node,scale=0.6,anchor=west] (leg4) at ([xshift=1cm, yshift=0.2cm]leg3txt.east) {};
\node[anchor=west] (leg4txt) at (leg4.east) {\textbf{visited} \& \textsc{occupied}};

\node[part,node,scale=0.6,anchor=north] (leg2) at ([yshift=-0.2cm]leg1.south) {};
\node[anchor=west] (leg2txt) at (leg2.east) {\textbf{partiallyVisited} \& \textsc{vacated}};

\node[full,node,scale=0.6,anchor=north] (leg5) at ([yshift=-0.2cm]leg4.south) {};
\node[anchor=west] (leg5txt) at (leg5.east) {\textbf{fullyVisited} \& \textsc{vacated}};

\path (leg5txt.east) -- (leg4txt.east) coordinate[midway] (rightalign);
\node[draw,rounded corners,inner sep=4pt,fit=(leg1)(leg2)(leg3)(leg3txt)(leg4)(leg4txt)(leg5)(leg5txt)(rightalign)] {};

\end{tikzpicture}}
\fi
\caption{Examples for rules \textbf{(R1)}–\textbf{(R3c-ii)}.  
Neighbors of $x$ are labeled with the rule that determines whether they are \textsc{empty} or \textsc{occupied}.}
\label{fig:probe}
\end{figure}
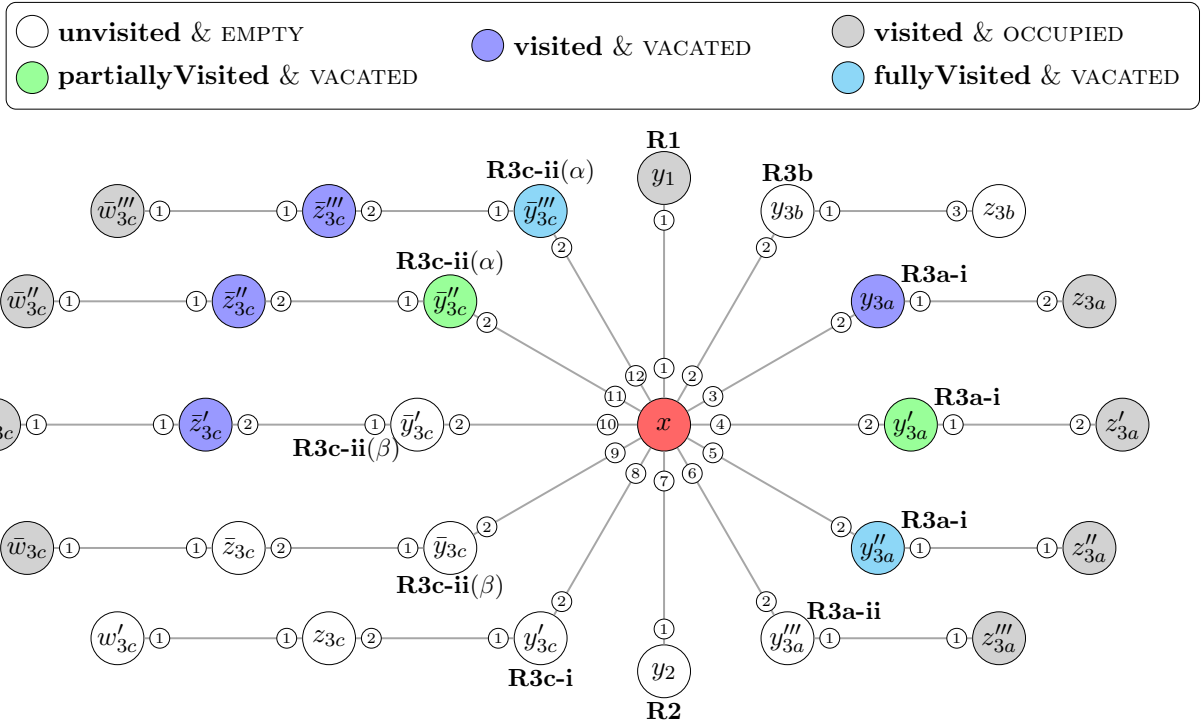

%% file: 11_rooted.tex
\section{Rooted Dispersion in \async}\label{sec:rooted}
The rooted dispersion of $k \le n$ agents unfolds in two phases. In the first phase, the agents follow Algorithm~\ref{alg:portoneDFS} (\texttt{DFS\_P1tree()}) to construct a {\sc P1Tree}, only until there are unsettled agents. Once the tree size is $k$, the construction is complete. The agents that were settled at \textsc{vacated} nodes have traveled with the DFShead helping with the construction of the {\sc P1Tree}. 
In the second phase, the agents retrace their path along the tree edges to return the settled agents from \textsc{vacated} nodes to their originally settled node.  
We will utilize the techniques described in Sections~\ref{subsec:vacant} and~\ref{subsec:probe} to construct the {\sc P1Tree}. 
Further, we also use Algorithm~\ref{alg:retrace} (\texttt{Retrace()}; pseudocode in Appendix~\ref{app:retrace}) to return the \emph{settledScout} agents to the original node they settled at.
Now, we describe the Algorithm~\ref{alg:rooted} (\texttt{RootedAsync()}; pseudocode in Appendix~\ref{app:rooted}). 

\vspace{2mm}
\noindent\textbf{Description.} Initially, the agents are located at $v_0$ and have state \emph{unsettled}. 
On visiting a node $v$, if there is no settled agent in $v$, the agent with the highest ID among the unsettled agents at $v$ settles and becomes $\psi(v)$. The settled agent sets its \textsf{parentPort} to $\bot$ and its \textsf{arrivalPort} to the port it used to reach $v$. 
It sets \textsf{portAtParent} to the port at the parent node that was used to reach $v$. 
All other agents at $v$ are part of the scout pool, denoted as $A_{scout}$. The settled agent conducts neighborhood search using \texttt{Parallel\_Probe()}. 
Then it checks if it can vacate its position by calling Algorithm~\ref{alg:vacant} \texttt{Can\_Vacate()} (pseudocode in Appendix \ref{app:vacant}). If the settled agent can vacate, it becomes part of $A_{vacated}$ and is added to $A_{scout}$ with state \emph{settledScout}. All agents in $A_{scout}$ move through the next port returned by \texttt{Parallel\_Probe()}. If no port is available, all agents in $A_{scout}$ move through $\psi(v).\textsf{parentPort}$. The settled agent updates its \textsf{recentPort} to the port it used to move. The process continues until no unsettled agents remain.
To keep track of the tree, a settled agent at $v$ also keeps track of its sibling $w$ (if exists) that was visited immediately before at its parent $u$ and the settled agent at $u$ only keeps track of the last child (i.e., $v$). 

\subsection{Retrace Phase}\label{subsec:retrace}
Once the {\sc P1Tree} contains exactly $k$ vertices, each \textsc{vacated} node
has a \emph{settledScout} that travelled alongside the DFShead during the
construction. These agents, collected in the set $A_{\mathit{vacated}}$ and
currently co-located at the last vertex added to the tree, now have to walk
backwards through the tree so that every \textsc{vacated} vertex regains its
original settled agent.  We call this controlled backward traversal
\emph{retrace}.

\vspace{2mm}
\noindent\textbf{Local variables maintained by every settled agent
$\psi(v)$.}
\begin{itemize}\setlength\itemsep{0.1em}
  \item $\psi(v).\textsf{parent}$: ID of agent at parent node $u$ $\psi(u).\textsf{ID}$(is ${\bot}$ for the root), and the port at $u$ that leads to $v$;
  \item $\psi(v).\textsf{parentPort}$: the port of $v$ that leads to $u$;
  \item $\psi(v).\textsf{recentChild}$: the port at $v$ that leads to the
        most recently visited child;
  \item $\psi(v).\textsf{sibling}$: the ID, and the port of a child at $u$ which was visited immediately before $v$. It is $\bot$ when $v$ is the first child.
\end{itemize}
These pointers are sufficient for retrace:  
\emph{recentChild} tells us where the DFS went last,
and \emph{sibling} lets us jump sideways to the child that was explored
immediately before the current one.

\vspace{2mm}
\noindent\textbf{Description.}  
Retrace performs a depth first \emph{post-order} walk of the tree:
whenever it backtracks to an internal vertex~$u$ from node $v$ via the port
$\psi(u).\textsf{recentChild}$, it first tries to move
sideways to previous sibling of $v$ (if any, stored in $\psi(v).\textsf{sibling}$); only when no such sibling
exists, it ascends further to the parent of~$u$.
Each sideways jump also updates $\psi(u).\textsf{recentChild}$ so that the just-visited
leaf is \emph{logically deleted} from the tree.
An agent from $A_{vacated}$ realizes that it has reached its original node from the ID stored at a settled agent that we keep track as the next agent ID. Then the agent changes its state from \emph{settledScout} to \emph{settled}. Once all agents become \emph{settled}, \texttt{Retrace()} ends. {\dis} is achieved.

A small example run of Algorithm~\ref{alg:rooted} is in Section~\ref{sec:example} and Table~\ref{tab:variable} lists all the variables used by the agents in the Appendix.

\subsection{Correctness and Complexity}
\begin{lemma}\label{lem:scoutsize}
    For any parallel probe in Algorithm~\ref{alg:rooted} (\emph{\texttt{RootedAsync()}}), $|A_{scout}| \geq \lceil (k-2)/3\rceil$.
\end{lemma}
\begin{proof}
    Parallel Probe happens only when there are unsettled agents at the current position of DFShead.
    Initially, the size of $\mathcal{T}$ is 1, and there are $k-1$ unsettled agents in $A_{scout}$; thus $k-1 > \lceil (k-2)/3\rceil$.
    Consider a tree $\mathcal{T}$ of size $j$ constructed, and the DFShead is at the $j+1$th vertex. 
    From Lemma~\ref{lem:vacated}, there are at least $\lfloor j/3 \rfloor$ vacated nodes, thus $\lfloor j/3 \rfloor$ agents in $A_{vacated}$. The current position of DFShead has one settled agent and $k - (j+1)$ unsettled agents. Since $k - (j+1) \ge 1$, we have $j \leq k -2$.
    So, $|A_{scout}| = |A_{vacated}| + |A_{unsettled}| =  \lfloor j/3 \rfloor + k - (j +1)$. For $j < k -1$, $\lfloor j/3 \rfloor + k - 1 -j \ge \lceil (k-2)/3 \rceil$.
\end{proof}

We have the following remark due to Lemmas~\ref{lemma:parallelprobe} and \ref{lem:scoutsize}. 
\begin{remark}
    At each node where DFShead performs \texttt{Parallel\_Probe()}, it takes at most $18 = O(1)$ epochs.
\end{remark}
Since DFShead performs \texttt{Parallel\_Probe()} only when there are unsettled agents, and it needs to check at most $k-1$ ports at the root node and $k-2$ ports (excluding the parent port) at a non-root node. Since there are $k-1$ scouts at root node and at least $\lceil (k-2)/3\rceil$ scouts at other nodes, \texttt{Parallel\_Probe()} runs for at most three iterations, each of which is at most 6 epochs.

\begin{lemma}\label{lem:retrace}
    Algorithm~\ref{alg:retrace} (\emph{\texttt{Retrace()}}) takes at most $O(k)$ epochs.
\end{lemma}
\begin{proof}
    Whenever retrace returns to a vertex~$u$ via its most recent child,
either (i) a sibling exists, in which case that sibling becomes the new
\textsf{recentChild} and is visited next, or (ii) no sibling exists, so
\textsf{recentChild} is cleared and retrace ascends to~$u$'s parent.
Thus each edge is visited exactly once in the reverse order of its
creation, guaranteeing that every \textsc{vacated} vertex regains its agent
and that the walk terminates at the root.

Since each tree edge is traversed at most twice (once sideways or downwards
and once while backtracking), the total number of moves made by the agents in
$A_{vacated}$ is $O(k)$.  No additional memory beyond the two local
pointers per settled vertex is required.
\end{proof}

\begin{theorem}
    Algorithm~\ref{alg:rooted} (\emph{\texttt{RootedAsync()}}) takes at most $O(k)$ epochs to achieve {\dis} with $O(\log k+\Delta)$ bits of memory at each agent.
\end{theorem}
\begin{proof}
    Algorithm~\ref{alg:rooted} (\texttt{RootedAsync()}) performs at most $k-1$ forward phases to settle $k-1$ unsettled agents at the root. Each forward phase consists of \texttt{Parallel\_Probe()}, which takes $O(1)$ epochs for determining the next node to visit. Algorithm \texttt{Can\_Vacate()} takes at most $4 = O(1)$ epochs every time an agent is settled to determine if it needs to be vacated; due to DFShead moving to its port-1 neighbor and parent node in the worst case.
    When $\delta_x$ at a node $x$ is larger than $k - 2$, the agents must find sufficient empty neighbors to settle all the unsettled agents. Hence it takes at most $O(1)$ epochs to settle all unsettled agents.
    Finally, the agents that belong to \textsc{vacated} nodes return to their home nodes in $O(k)$ epochs as per Lemma~\ref{lem:retrace}. 
    Hence, overall Algorithm~\ref{alg:rooted} (\texttt{RootedAsync()}) runs in $O(k)$ epochs.

    For the memory complexity, notice that each agent stores a constant number of port numbers (corresponding to probing and retrace) of size $O(\log \Delta)$ bits and IDs (one of each for probe target, port-1 neighbor of probe target, child, sibling and parent) of size $O(\log k)$ bits.
    Other information such as state and types are $O(1)$ bits each.
    Hence in total it needs $O(\log (k+\Delta))$ bits to store all the variables.
\end{proof}

%% file: 2_general.tex
\sloppy
\section{General Dispersion}\label{sec:general}
The idea of general dispersion, i.e., when there are multiple nodes in the initial configuration with more than one agent, broadly follows from the merger strategy of Kshemkalyani and Sharma \cite{KshemkalyaniOPODIS21}. 
Each multiplicity starts its own \texttt{RootedAsync()} with a treelabel that consists of $\langle a_{highest}.\textsf{ID}, a_{highest}.\textsf{level}, a_{highest}.\textsf{weight} \rangle$ to create its {\sc P1Tree}. 
The \textsf{level} parameter starts at $0$. The level increases for every merger between two \textsc{P1Tree} of the same size. Otherwise, the level remains the same and the higher weight \textsc{P1Tree} wins. The weight of the traversal is updated to reflect the total number of agents in the merged \textsc{P1Tree}.

We use the following modifications to make the \texttt{RootedAsync()} compatible with the merger. When doing \texttt{Parallel\_Probe()}, a node $y$ is considered to be \textsc{empty} if the treelabel of $\xi(y)$ is different from the treelabel of the scouting agent.
Analogously for \texttt{Can\_Vacate()} when going to the port-1 neighbor.

Every time a lower priority traversal meets a higher priority traversal, it revisits all the nodes of its own traversal to collect all the settled agents and joins the higher priority traversal, by chasing the DFShead. If a higher priority traversal finds a lower priority traversal head, then it subsumes it completely, otherwise, if it finds settled agents, then it treats those nodes as empty nodes and absorbs the settled agent.
To facilitate the collection of agents by a lower priority DFShead, the agents essentially perform \texttt{Retrace()} to traverse the tree, and unlike retrace, all settled agents move with the DFShead. Once all agents are collected at the root, they follow \textsf{recentPort} to reach the previous position of DFShead where higher priority DFS tree exists.

With these primitives added to \texttt{RootedAsync()}, it can utilize the merger strategy of Kshemkalyani and Sharma \cite{KshemkalyaniOPODIS21} with overhead proportional to size of the tree, which is $O(k)$ in the worst case. Thus we achieve a $O(k)$ time solution for any arbitrary distribution of agents on the graph.

%% file: example.tex
\begin{toappendix}
\section{An Example Run of Algorithm \texttt{RootedAsync()}}
\label{sec:example}
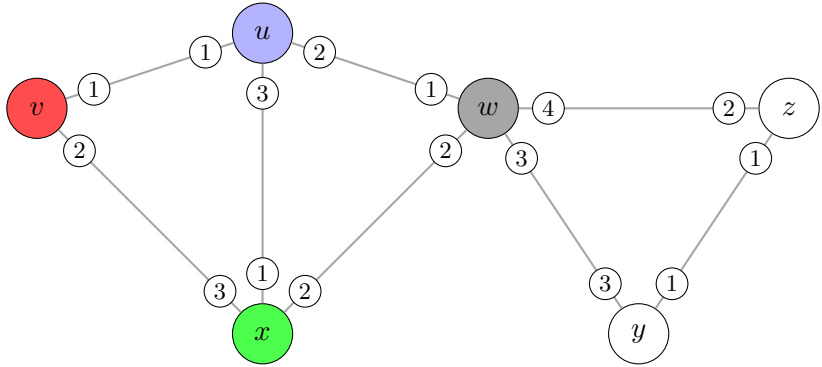
\begin{figure}[ht]
    \centering
    \begin{tikzpicture}[
    every node/.style   = {circle, draw=black, minimum size=8mm, text=black},
    port/.style         = {circle, fill=white, draw=black, text=black,
                           font=\footnotesize, inner sep=1pt, minimum size=12pt},
    edge/.style         = {draw=gray!70, thick}
]

\node[fill=red!70]   (v0) at ( 0, 0)   {$v$};
\node[fill=blue!30]  (v1) at ( 3, 1)   {$u$};
\node[fill=gray!70]  (v2) at ( 6, 0)   {$w$};
\node[fill=green!70] (v3) at ( 3, -3)  {$x$};
\node[fill=white]  (v4) at (8, -3)  {$y$};
\node[fill=white]  (v5) at (10, 0)  {$z$};

\foreach \src/\dst/\psrc/\pdst in {
    v0/v1/1/1, 
    v1/v2/2/1, 
    v2/v3/2/2, 
    v3/v1/1/3, 
    v3/v0/3/2, 
    v2/v4/3/3, 
    v4/v5/1/1, 
    v5/v2/2/4  
}{
    \draw[edge] (\src) -- (\dst);
    \node[port] at ($(\src)!0.8cm!(\dst)$) {\psrc};
    \node[port] at ($(\dst)!0.8cm!(\src)$) {\pdst};
}
\end{tikzpicture}
    \caption{Graph for the example run. Node colors are illustrative from the original image and may not reflect dynamic states in this trace.}
    \label{fig:probing_example_corrected_order}
\end{figure}

Consider $k=6$ agents, $a_1, a_2, a_3, a_4, a_5, a_6$, with IDs sorted in ascending order (i.e., $a_1.\textsf{ID} < a_2.\textsf{ID} < \dots < a_6.\textsf{ID}$). Initially, all agents are at node $v$ (the root) and are in state \textit{unsettled}. The DFS-based P1Tree construction begins. $\psi(N)$ denotes the agent settled at node $N$. $A_{unsettled}$ is the set of unsettled agents. $A_{vacated}$ is the set of agents whose nodes were vacated. $A_{scout} = A_{unsettled} \cup A_{vacated}$. The agent $a_{min}$ is the agent in $A_{scout}$ with the lowest ID.

\begin{enumerate}
    \item \textbf{Settle at $v$ (Root):}
    Agents $\{a_1, \dots, a_6\}$ are at $v$.
    $a_6$ (highest ID) settles: $\psi(v) = a_6$. $a_6.state \gets \textit{settled}$.
    $A_{unsettled} = \{a_1, \dots, a_5\}$. $A_{scout} = \{a_1, \dots, a_5\}$. $a_{min} = a_1$.
    $\psi(v).\textsf{parentPort} = \bot$.
    
    \textbf{\texttt{Parallel\_Probe()} at $v$:} Ports are 1 (to $u$) and 2 (to $x$).
    $a_1$ scouts port 1 (to $u$); $a_2$ scouts port 2 (to $x$).
    \begin{itemize}
        \item $a_1$ (to $u$): Edge $\{v,u\}$ is type \texttt{t11}. Node $u$ is \textbf{unvisited}. Scout $a_1$ finds $\xi(u)=\bot$ and $p_{uv}=1$ (Rule R2). Reports $u$ as \textsc{empty}.
        \item $a_2$ (to $x$): Edge $\{v,x\}$ is type \texttt{tpq}. Node $x$ is \textbf{unvisited}. Scout $a_2$ finds $\xi(x)=\bot, p_{xv} \neq 1$. (Rule R3).
        port-1 neighbor (P1N) of $x$ is $u$. $a_2$ visits $u$. $\xi(u)=\bot, p_{ux} \neq 1$. (Rule R3c).
        P1N of $u$ is $v$. $a_2$ visits $v$. $\xi(v)=a_6$. Reports $x$ as \textsc{empty}.
    \end{itemize}
    Probe results processed: $u$ (\texttt{t11}, \textsc{empty}), $x$ (\texttt{tpq}, \textsc{empty}).
    Based on probe results, node $v$ is \textbf{visited} (it has unvisited neighbors).
    
    \texttt{Can\_Vacate($\psi(v)=a_6$)}: $\psi(v).\textsf{parentPort} = \bot$ (root). $a_6$ remains \textit{settled}. Node $v$ is \textsc{occupied}.
    Next edge is to $u$ (priority). $A_{scout}$ moves to $u$. $a_6.\textsf{recentChild} \gets (\text{port 1 to } u)$. $a_{min}.\textsf{childPort}$ (i.e. $a_1.\textsf{childPort}$) $\gets (\text{port 1 to } u)$.
    
    \item \textbf{Settle at $u$:}
    $a_5$ (highest ID in $A_{unsettled}=\{a_1, \dots, a_5\}$) settles: $\psi(u) = a_5$. $a_5.state \gets \textit{settled}$.
    $a_5.\textsf{parent} \gets (a_6.\textsf{ID}, \text{port 1 at } v)$. $a_5.\textsf{parentPort} \gets (\text{port 1 at } u)$.
    $A_{unsettled} = \{a_1, \dots, a_4\}$.
    
    \textbf{\texttt{Parallel\_Probe()} at $u$:} Ports (excl. parent port 1) are 2 (to $w$), 3 (to $x$).
    $a_1$ scouts port 2 (to $w$); $a_2$ scouts port 3 (to $x$).
    \begin{itemize}
        \item $a_1$ (to $w$): Edge $\{u,w\}$ is type \texttt{tp1}. Node $w$ is \textbf{unvisited}. Rule R2. Reports $w$ as \textsc{empty}.
        \item $a_2$ (to $x$): Edge $\{u,x\}$ is type \texttt{tp1}. Node $x$ is \textbf{unvisited}. Rule R2. Reports $x$ as \textsc{empty}.
    \end{itemize}
    Probe results: $w$ (\texttt{tp1}, \textsc{empty}), $x$ (\texttt{tp1}, \textsc{empty}).
    Node $u$ is \textbf{visited}.

    \texttt{Can\_Vacate($\psi(u)=a_5$)}: $\psi(u).\textsf{nodeType} = \textbf{visited}$. P1N of $u$ is $v$. $\xi(v)=a_6 \neq \bot$ (and $v$ is \textsc{occupied}). Rule (V2) applies. $a_5$ becomes \textit{settledScout}. Node $u$ is \textsc{vacated}.
    $A_{vacated} = \{a_5\}$. $A_{scout} = \{a_1, \dots, a_4, a_5\}$. $a_{min}=a_1$.
    $\psi(u)=a_5$ stores $a_5.\textsf{P1Neighbor} \gets a_6.\textsf{ID}$, $a_5.\textsf{portAtP1Neighbor} \gets p_{vu}=1$.
    Next edge to $w$. $A_{scout}$ moves to $w$. $\psi(u)=a_5$ (now scout) updates $a_5.\textsf{recentChild} \gets (\text{port 2 to } w)$. $a_1.\textsf{childPort} \gets (\text{port 2 to } w)$.
    
    \item \textbf{Settle at $w$:}
    $a_4$ settles: $\psi(w) = a_4$. $a_4.state \gets \textit{settled}$.
    $a_4.\textsf{parent} \gets (a_5.\textsf{ID}, \text{port 2 at } u)$. $a_4.\textsf{parentPort} \gets (\text{port 1 at } w)$.
    $A_{unsettled} = \{a_1, a_2, a_3\}$.
    
    \textbf{\texttt{Parallel\_Probe()} at $w$:} Ports (excl. parent port 1) are 2 (to $x$), 3 (to $y$), 4 (to $z$).
    $a_1$ scouts port 2 (to $x$); $a_2$ scouts port 3 (to $y$); $a_3$ scouts port 4 (to $z$).
    \begin{itemize}
        \item $a_1$ (to $x$): Reports $x$ as \textsc{empty} (as per step 2b logic, after 2-hop check, finds no scout $\psi(x)$ or $\psi(u)$ in $A_{scout}$ that are settled at those nodes for P1N purposes, as $a_5=\psi(u)$ is a scout now, but its P1N info is about $v$).
        \item $a_2$ (to $y$): Reports $y$ as \textsc{empty} (Rule R3b, P1N $z$ is empty, $p_{zy}=1$).
        \item $a_3$ (to $z$): Reports $z$ as \textsc{empty} (Rule R3b, P1N $y$ is empty, $p_{yz}=1$).
    \end{itemize}
    Probe results: $x, y, z$ all (\texttt{tpq}, \textsc{empty}).
    Node $w$ is \textbf{visited}.

    \texttt{Can\_Vacate($\psi(w)=a_4$)}: $\psi(w).\textsf{nodeType} = \textbf{visited}$. P1N of $w$ is $u$. $\xi(u)=\bot$ (since $a_5=\psi(u)$ is a scout). Condition ``P1N is \textsc{occupied}'' is false. $a_4$ remains \textit{settled}. Node $w$ is \textsc{occupied}.
    $A_{vacated} = \{a_5\}$. $A_{scout} = \{a_1, a_2, a_3, a_5\}$. $a_{min}=a_1$.
    Next edge to $x$. $A_{scout}$ moves to $x$. $a_4.\textsf{recentChild} \gets (\text{port 2 to } x)$. $a_1.\textsf{childPort} \gets (\text{port 2 to } x)$.
    
    \item \textbf{Settle at $x$:}
    $a_3$ settles: $\psi(x) = a_3$. $a_3.state \gets \textit{settled}$.
    $a_3.\textsf{parent} \gets (a_4.\textsf{ID}, \text{port 2 at } w)$. $a_3.\textsf{parentPort} \gets (\text{port 2 at } x)$.
    $A_{unsettled} = \{a_1, a_2\}$. $A_{scout} = \{a_1, a_2, a_5\}$. $a_{min}=a_1$.
    
    \textbf{\texttt{Parallel\_Probe()} at $x$:} Ports (excl. parent port 2) are 1 (to $u$), 3 (to $v$).
    \begin{itemize}
        \item $a_1$ (to $u$): Edge $\{x,u\}$ is \texttt{t1p}. $\xi(u)=\bot$. P1N of $u$ is $v$. $\xi(v)=a_6$. At $x$, find $b=\psi(u)$, which is $a_5 \in A_{scout}$, with $a_5.\textsf{P1Neighbor}=a_6.\textsf{ID}$. Yes. Reports $u$ as \textsc{vacated} (\textbf{visited}, Rule R3a-ii).
        \item $a_2$ (to $v$): Edge $\{x,v\}$ is \texttt{tpq}. $\xi(v)=a_6$. Rule R1. Reports $v$ as \textsc{occupied} (\textbf{visited}).
    \end{itemize}
    Probe results: All explorable non-parent neighbors ($u,v$) are not \textsc{empty}. Parent edge $\{w,x\}$ is \texttt{tpq}. Assume no other \textsc{empty} neighbors via non-\texttt{tpq} edges.
    Node $x$ becomes \textbf{partiallyVisited}.

    \texttt{Can\_Vacate($\psi(x)=a_3$)}: $\psi(x).\textsf{nodeType} = \textbf{partiallyVisited}$. Rule (V4). $a_3$ becomes \textit{settledScout}. Node $x$ is \textsc{vacated}.
    $A_{vacated} = \{a_5, a_3\}$. $A_{scout} = \{a_1, a_2, a_5, a_3\}$.
    DFS backtracks (nextPort is $\bot$). $A_{scout}$ moves to $w$. $a_1.\textsf{childDetails} \gets (a_3.\textsf{ID}, \text{port 2 at } w)$.
    
    \item \textbf{At $w$ (after backtrack from $x$):}
    $\psi(w)=a_4$.
    \textbf{\texttt{Parallel\_Probe()} at $w$:} Next ports to probe: 3 (to $y$), 4 (to $z$).
    \begin{itemize}
        \item $a_1$ scouts port 3 (to $y$): reports \textsc{empty}. Same as before.
        \item $a_2$ scouts port 4 (to $z$): reports \textsc{empty}.  Same as before.
    \end{itemize}
    Probe results update: $y$ and $z$ are \textsc{empty}. Node $w$ remains \textbf{visited}.

    \texttt{Can\_Vacate($\psi(w)=a_4$)}: No change, $a_4$ remains \textit{settled}, $w$ is \textsc{occupied}.
    Priority to $y$ (port 3). $A_{scout}$ moves to $y$.
    $a_4.\textsf{recentChild} \gets (\text{port 3 to } y)$. $a_1.\textsf{childPort} \gets (\text{port 3 to } y)$.
    $a_1.\textsf{siblingDetails} \gets (a_3.\textsf{ID}, \text{port 2 at } w)$ (sibling of $y$ is $x$).
    
    \item \textbf{Settle at $y$:}
    $a_2$ settles: $\psi(y) = a_2$. $a_2.state \gets \textit{settled}$.
    $a_2.\textsf{parent} \gets (a_4.\textsf{ID}, \text{port 3 at } w)$. $a_2.\textsf{parentPort} \gets (\text{port 3 at } y)$.
    $a_2.\textsf{sibling} \gets (a_3.\textsf{ID}, \text{port 2 at } w)$.
    $A_{unsettled} = \{a_1\}$. $A_{scout} = \{a_1, a_5, a_3\}$. $a_{min}=a_1$.
    
    \textbf{\texttt{Parallel\_Probe()} at $y$:} Ports (excl. parent port 3) are 1 (to $z$), 2 (to $x$).
    \begin{itemize}
        \item $a_1$ (to $z$): Edge $\{y,z\}$ is \texttt{t11}. Node $z$ is \textbf{unvisited}. Rule R2. Reports $z$ as \textsc{empty}.
        \item $a_5$ (to $x$): Edge $\{y,x\}$ is \texttt{tpq}. $\xi(x)=\bot$. P1N $u$, P1N of $u$ is $v (\xi(v)=a_6)$.
        At $y$: find $c=\psi(u)=a_5 \in A_{scout}$ (yes). Find $b=\psi(x)=a_3 \in A_{scout}$ (yes). Rule R3c-ii($\alpha$). Reports $x$ as \textsc{vacated} (\textbf{partiallyVisited}).
    \end{itemize}
    Probe results: $z$ (\texttt{t11}, \textsc{empty}), $x$ (\texttt{tpq}, \textbf{partiallyVisited}).
    Node $y$ is \textbf{visited}.

    \texttt{Can\_Vacate($\psi(y)=a_2$)}: $\psi(y).\textsf{nodeType} = \textbf{visited}$. P1N of $y$ is $z$. $\xi(z)=\bot$ (it's \textsc{empty}). Returns \textit{settled}. Node $y$ is \textsc{occupied}.
    Next edge to $z$. $A_{scout}$ moves to $z$.
    $a_2.\textsf{recentChild} \gets (\text{port 1 to } z)$. $a_1.\textsf{childPort} \gets (\text{port 1 to } z)$.
    $a_1.\textsf{siblingDetails} \gets \bot$.
    
    \item \textbf{Settle at $z$:}
    $a_1$ settles: $\psi(z) = a_1$. $a_1.state \gets \textit{settled}$.
    $a_1.\textsf{parent} \gets (a_2.\textsf{ID}, \text{port 1 at } y)$. $a_1.\textsf{parentPort} \gets (\text{port 1 at } z)$.
    $a_1.\textsf{sibling} \gets \bot$.
    $A_{unsettled} = \{\}$. $k=6$ agents are settled.
\end{enumerate}
    Construction phase ends as $A_{unsettled}$ is empty.
    Current agent states and logical locations (physical location of scouts is $z$):
    $v: \psi(v)=a_6 (\textit{settled})$
    $u: \psi(u)=a_5 (\textit{settledScout})$
    $w: \psi(w)=a_4 (\textit{settled})$
    $x: \psi(x)=a_3 (\textit{settledScout})$
    $y: \psi(y)=a_2 (\textit{settled})$
    $z: \psi(z)=a_1 (\textit{settled})$
    $A_{vacated} = \{a_3, a_5\}$ (sorted by ID). They are physically co-located at $z$.
    
\textbf{Retrace Phase:} $A_{vacated} = \{a_3, a_5\}$. Current location of $A_{vacated}$ group: $z$. The lead retrace agent $a_{min\_retrace}$ is $a_3$.
    Goal: $a_3$ to $x$, $a_5$ to $u$.
    \begin{enumerate}
        \item[a.] \textbf{At $z$ (settled agent $\psi(z)=a_1$):}
        $A_{vacated}=\{a_3, a_5\}$ is present. $a_{min}=a_3$. Node $z$ is occupied by $a_1 (\xi(z)=a_1)$.
        $\psi(z)=a_1$ has $a_1.\textsf{recentChild} = \bot$ (leaf in construction).
        The group moves to parent of $z$.
        $a_3.\textsf{nextAgentID} \gets \psi(z).\textsf{parent.ID}$ (i.e., $a_2.\textsf{ID}$).
        $a_3.\textsf{nextPort} \gets \psi(z).\textsf{parentPort}$ (port 1 at $z$).
        $a_3.\textsf{siblingDetails} \gets \psi(z).\textsf{sibling}$ ($\bot$).
        $A_{vacated}=\{a_3, a_5\}$ moves to $y$ via $z$'s port 1. $a_3.\textsf{arrivalPort}$ at $y$ becomes port 1.
    
        \item[b.] \textbf{At $y$ (settled agent $\psi(y)=a_2$):}
        $A_{vacated}=\{a_3, a_5\}$ arrives. $a_{min}=a_3$. Node $y$ is occupied by $a_2 (\xi(y)=a_2)$.
        $\psi(y)=a_2$ has $a_2.\textsf{recentChild} = (\text{port 1 to } z)$.
        $a_3.\textsf{arrivalPort}$ (port 1 at $y$) matches $\psi(y).\textsf{recentChild}$.
        $a_3.\textsf{siblingDetails}$ is $\bot$.
        $\psi(y).\textsf{recentChild} \gets \bot$.
        The group moves to parent of $y$.
        $a_3.\textsf{nextAgentID} \gets \psi(y).\textsf{parent.ID}$ (i.e., $a_4.\textsf{ID}$).
        $a_3.\textsf{nextPort} \gets \psi(y).\textsf{parentPort}$ (port 3 at $y$).
        $a_3.\textsf{siblingDetails} \gets \psi(y).\textsf{sibling}$ ($(a_3.\textsf{ID}, \text{port 2 at } w)$).
        $A_{vacated}=\{a_3, a_5\}$ moves to $w$ via $y$'s port 3. $a_3.\textsf{arrivalPort}$ at $w$ becomes port 3.
    
        \item[c.] \textbf{At $w$ (settled agent $\psi(w)=a_4$):}
        $A_{vacated}=\{a_3, a_5\}$ arrives. $a_{min}=a_3$. Node $w$ is occupied by $a_4 (\xi(w)=a_4)$.
        $\psi(w)=a_4$ has $a_4.\textsf{recentChild} = (\text{port 3 to } y)$.
        $a_3.\textsf{arrivalPort}$ (port 3 at $w$) matches $\psi(w).\textsf{recentChild}$.
        $a_3.\textsf{siblingDetails}$ is $(a_3.\textsf{ID}, \text{port 2 at } w)$, which is not $\bot$.
        The group moves to the sibling node $x$.
        $a_3.\textsf{nextAgentID} \gets a_3.\textsf{ID}$ (from siblingDetails).
        $a_3.\textsf{nextPort} \gets (\text{port 2 at } w)$ (from siblingDetails).
        $\psi(w).\textsf{recentChild} \gets (\text{port 2 at } w)$ (updates to current traversal direction).
        $a_3.\textsf{siblingDetails} \gets \bot$.
        $A_{vacated}=\{a_3, a_5\}$ moves to $x$ via $w$'s port 2. $a_3.\textsf{arrivalPort}$ at $x$ becomes port 2.
    
        \item[d.] \textbf{At $x$ (original node of $a_3$, currently \textsc{vacated}):}
        $A_{vacated}=\{a_3, a_5\}$ arrives. $a_{min}=a_3$. Node $x$ is \textsc{vacated} ($\xi(x)=\bot$).
        $a_3.\textsf{nextAgentID}$ is $a_3.\textsf{ID}$. Agent $a_3 \in A_{vacated}$ matches.
        $a_3.state \gets \textit{settled}$. $a_3$ occupies $x$. $\psi(x) \gets a_3$.
        $A_{vacated}$ becomes $\{a_5\}$. $a_{min}$ (for the remaining $A_{vacated}$) is now $a_5$.
        $\psi(x)=a_3$ has $a_3.\textsf{recentChild} = \bot$.
        The group (now just $\{a_5\}$) moves to parent of $x$.
        $a_5.\textsf{nextAgentID} \gets \psi(x).\textsf{parent.ID}$ (i.e., $a_4.\textsf{ID}$).
        $a_5.\textsf{nextPort} \gets \psi(x).\textsf{parentPort}$ (port 2 at $x$).
        $a_5.\textsf{siblingDetails} \gets \psi(x).\textsf{sibling} (\bot)$.
        $A_{vacated}=\{a_5\}$ moves to $w$ via $x$'s port 2. $a_5.\textsf{arrivalPort}$ at $w$ becomes port 2.
    
        \item[e.] \textbf{At $w$ (settled agent $\psi(w)=a_4$):}
        $A_{vacated}=\{a_5\}$ arrives. $a_{min}=a_5$. Node $w$ is occupied by $a_4 (\xi(w)=a_4)$.
        $\psi(w)=a_4$ has $a_4.\textsf{recentChild} = (\text{port 2 at } w)$ (updated in step c).
        $a_5.\textsf{arrivalPort}$ (port 2 at $w$) matches $\psi(w).\textsf{recentChild}$.
        $a_5.\textsf{siblingDetails}$ is $\bot$.
        $\psi(w).\textsf{recentChild} \gets \bot$.
        The group moves to parent of $w$.
        $a_5.\textsf{nextAgentID} \gets \psi(w).\textsf{parent.ID}$ (i.e., $a_5.\textsf{ID}$).
        $a_5.\textsf{nextPort} \gets \psi(w).\textsf{parentPort}$ (port 1 at $w$).
        $a_5.\textsf{siblingDetails} \gets \psi(w).\textsf{sibling} (\bot)$.
        $A_{vacated}=\{a_5\}$ moves to $u$ via $w$'s port 1. $a_5.\textsf{arrivalPort}$ at $u$ becomes port 2.
    
        \item[f.] \textbf{At $u$ (original node of $a_5$, currently \textsc{vacated}):}
        $A_{vacated}=\{a_5\}$ arrives. $a_{min}=a_5$. Node $u$ is \textsc{vacated} ($\xi(u)=\bot$).
        $a_5.\textsf{nextAgentID}$ is $a_5.\textsf{ID}$. Agent $a_5 \in A_{vacated}$ matches.
        $a_5.state \gets \textit{settled}$. $a_5$ occupies $u$. $\psi(u) \gets a_5$.
        $A_{vacated}$ becomes $\emptyset$.
        Retrace phase ends as $A_{vacated}$ is empty.
    \end{enumerate}
    
Final agent settlement: $\psi(v)=a_6, \psi(u)=a_5, \psi(w)=a_4, \psi(x)=a_3, \psi(y)=a_2, \psi(z)=a_1$. Dispersion is achieved.
\end{toappendix}

%% file: tablevariable.tex
\begin{toappendix}
\section{Table of Variables}\label{sec:tablevariable}
\begin{longtable}{@{}>{\ttfamily}l p{0.7\textwidth}@{}}
\caption{Variables Used by Agents}\label{tab:variable}\\
\toprule
\textbf{Variable Name} & \textbf{Description} \\
\midrule
\endfirsthead
\multicolumn{2}{c}%
{{\tablename\ \thetable{} -- continued from previous page}} \\
\toprule
\textbf{Variable Name} & \textbf{Description} \\
\midrule
\endhead
\bottomrule
\multicolumn{2}{r}{{Continued on next page}} \\
\endfoot
\bottomrule
\endlastfoot

\multicolumn{2}{@{}l}{\textbf{Generic Agent Properties (applicable to any agent $a$)}} \\ \addlinespace
$a.\textsf{ID}$ & Unique identifier of the agent. \\
$a.\textsf{state}$ & Current operational state of the agent (e.g., \textit{unsettled}, \textit{settled}, \textit{settledScout}). \\
$a.\textsf{arrivalPort}$ & The port number through which agent $a$ arrived at its current node. \\
$a.\textsf{treeLabel}$ & For general dispersion: a tuple $\langle \texttt{leaderID}, \texttt{level}, \texttt{weight} \rangle$ identifying the P1Tree exploration the agent is part of. Contains the ID of the tree's root agent, the tree's merger level, and its current weight (number of agents). \\
\addlinespace

\multicolumn{2}{@{}l}{\textbf{Variables for Settled Agents (e.g., agent $a = \psi(v)$ at node $v$)}} \\ \addlinespace
$a.\textsf{nodeType}$ & The type of node $v$ where the agent is settled (e.g., \textbf{unvisited}, \textbf{partiallyVisited}, \textbf{fullyVisited}, \textbf{visited}). \\
$a.\textsf{parent}$ & A tuple: (ID of the agent settled at $v$'s parent node in the P1Tree, port number at the parent node leading to $v$). Is $\bot$ for the root agent. \\
$a.\textsf{parentPort}$ & The port number at node $v$ that leads to its parent in the P1Tree. Is $\bot$ for the root agent. \\
$a.\textsf{P1Neighbor}$ & ID of the agent settled at the port-1 neighbor of node $v$. Stores $\bot$ if the port-1 neighbor is \textsc{empty} or unvisited. \\
$a.\textsf{portAtP1Neighbor}$ & The port number at $v$'s port-1 neighbor (say $w$) that leads back to $v$. \\
$a.\textsf{vacatedNeighbor}$ & Boolean flag. True if a neighbor of $v$ (for which $v$ is a port-1 neighbor and would need to be \textsc{occupied} for that neighbor to be \textsc{vacated}) has itself become \textsc{vacated}. Used in Algorithm \texttt{Can\_Vacate}. \\
$a.\textsf{recentChild}$ & The port number at node $v$ that leads to the child most recently visited by the DFS traversal originating from $v$. \\
$a.\textsf{sibling}$ & A tuple: (ID of the agent at the previous sibling node in the DFS tree, port number at $v$'s parent leading to that sibling). Is $\bot$ if $v$ is the first child. \\
$a.\textsf{recentPort}$ & The port number most recently used by the scout agents to depart from node $v$ (either towards a child or back to the parent). \\
$a.\textsf{probeResult}$ & Stores the overall highest priority result (e.g., next edge to traverse) obtained from the \texttt{Parallel\_Probe} procedure executed at node $v$. \\
$a.\textsf{checked}$ & The count of incident ports at node $v$ that have already been explored during the \texttt{Parallel\_Probe} procedure. \\
\addlinespace

\multicolumn{2}{@{}l}{\textbf{Temporary Variables for Scouting Agents (agent $a \in A_{\textsc{scout}}$ during \texttt{Parallel\_Probe})}} \\ \addlinespace
$a.\textsf{scoutPort}$ & The port number at the current DFS head node that this scout agent $a$ is assigned to explore. \\
$a.\textsf{scoutEdgeType}$ & The type of the edge (e.g., \texttt{tp1}, \texttt{t11}) discovered by scout $a$ along its \texttt{scoutPort}. \\
$a.\textsf{scoutP1Neighbor}$ & (During probe of neighbor $y$) Stores ID of the agent at port-1 neighbor of $y$ (say $z$), or $\bot$. \\
$a.\textsf{scoutPortAtP1Neighbor}$ & (During probe of $y$) Stores port at $z$ leading to $y$. \\
$a.\textsf{scoutP1P1Neighbor}$ & (During probe of $y$'s P1N $z$) Stores ID of agent at port-1 neighbor of $z$ (say $w$), or $\bot$. \\
$a.\textsf{scoutPortAtP1P1Neighbor}$ & (During probe of $z$) Stores port at $w$ leading to $z$. \\
$a.\textsf{scoutResult}$ & A tuple $\langle p_{xy}, \texttt{edgeType}, \texttt{nodeType}_y, a' \rangle$ storing the individual result found by scout $a$ for its assigned port. \\
\addlinespace

\multicolumn{2}{@{}l}{\textbf{Context Variables for the Lead Scout Agent (e.g., $a = a_{\min}$)}} \\ \addlinespace
$a.\textsf{prevID}$ & ID of the agent settled at the node from which the DFS head (and scout group) just departed. \\
$a.\textsf{childPort}$ & Port at the current DFS head that will be taken to visit the next child. This info is used to set up the child's parent information. \\
$a.\textsf{siblingDetails}$ & A tuple carrying information about the current child's previous sibling, to be passed to the agent settling at the next child. Format: (Sibling Agent ID, Port at Parent to Sibling). \\
$a.\textsf{childDetails}$ & A tuple carrying information about the child node just exited during a backtrack operation, to be used by the parent. Format: (Child Agent ID, Port at Parent to Child). \\
$a.\textsf{nextAgentID}$ & (During Retrace phase) The ID of the agent whose original settled node the $A_{\textsc{vacated}}$ group is currently moving towards. \\
$a.\textsf{nextPort}$ & (During Retrace phase) The port number the $A_{\textsc{vacated}}$ group will take to reach the node associated with \texttt{nextAgentID}. \\

\end{longtable}
\end{toappendix}

%% file: 3_pseudocodes.tex
\begin{toappendix}
\section{Pseudocodes of Algorithms}
\subsection{Pseudocode of Algorithm \ref{alg:portonecentral} \texttt{Centralized\_P1Tree}()}\label{app:centralized}
\begin{algorithm}[H]
\caption{\textsc{Centralized\_P1Tree}$(G)$}
\label{alg:portonecentral}
\KwIn{connected, port-labelled graph $G=(V,E)$}
\KwOut{a Port-One tree $\mathcal{T}$ of $G$}
$\mathcal{T} \leftarrow \emptyset$\;
mark all vertices \textbf{unvisited}\;
\While{there exists an \textbf{\em unvisited} vertex}{
    pick any \textbf{unvisited} vertex $v$, push $v$ on a stack $S$\;
    \While{$S \neq \emptyset$}{
        $u \leftarrow$ pop $S$\;
        \ForEach{incident edge $e=[u,p_{uv},p_{vu},v]$ of type \emph{\texttt{tp1}, \texttt{t11} \emph{or} \texttt{t1q}}}{
            \If{$v$ is \textbf{\em unvisited}}{
                $\mathcal{T} \leftarrow \mathcal{T} \cup \{e\}$\;
                push $v$ on $S$\;
                mark $v$ \textbf{visited}\;
            }
        }
    }
}
sort all edges of type \texttt{tpq} in lexicographical order\;
\ForEach{edge $e$ in sorted order}{
    \If{$\mathcal{T} \cup \{e\}$ is acyclic}{
        $\mathcal{T} \leftarrow \mathcal{T} \cup \{e\}$\;
        \If{$\mathcal{T}$ forms a single connected component}{
            \textbf{break}\;
        }
    }
}
\Return $\mathcal{T}$\;
\end{algorithm}
\vfill
\pagebreak
\subsection{Pseudocode of Algorithm \ref{alg:portoneDFS} \texttt{DFS\_P1Tree}()}\label{app:distributed}

\begin{algorithm}[H]
    \caption{\texttt{DFS\_P1Tree()}}\label{alg:portoneDFS}
    \KwIn{Root vertex $v_0$, port-labelled graph $G = (V,E)$}
    \KwOut{{\sc P1Tree} $\mathcal{T}$}
    
    edge priority: $\texttt{tp1} \succ \texttt{t11} \sim \texttt{t1q} \succ \texttt{tpq}$, smallest incident port number under each type\; \label{alg:portoneDFS:priority}
    
    initialize: $\mathcal{T} \leftarrow \emptyset$, mark all vertices \textbf{unvisited}, stack $S \gets \emptyset$\;
    $S.push(v_0)$\;
    $\texttt{type}(v_0) \gets \textbf{visited}$\;
    \While{$S \neq \emptyset$}{ \label{alg:portoneDFS:while}
        $u \leftarrow S.top()$\; \label{alg:portoneDFS:top}
        $e_{next} \gets \varnothing$\;
        $\mathcal{E} \gets $ sorted list of edges incident to $u$ in order of edge-priority\;
        \For{$e \in \mathcal{E}$ }{
            let $e = [u,p_{uv}, p_{vu}, v]$ be the edge\;
            \eIf{\em $\texttt{type}(v) = \textbf{unvisited}$}{
            $e_{next} \gets e$\;
            break\;
            }
            {
                \If{\em \texttt{type}$(v)$ = \textbf{partiallyVisited} and $\texttt{type}(\{u,v\}) \in \{\texttt{tp1}, \texttt{t11}\}$}
                {
                    $e_{next} \gets e$\;
                    break\;
                }
            }
        }
        \eIf{$e_{next} \neq \varnothing$}{ \label{alg:portoneDFS:ifEdge}
            let $e=[u,p_{uv},p_{vu},v]$ be the edge\; \label{alg:portoneDFS:letEdge}
            let $e_{\uparrow}=[w,p_{wu},p_{uw},u]$ be the parent edge of $u$\; \label{alg:portoneDFS:parentEdge}
            \eIf{$e,e_{\uparrow}$ are \texttt{tpq} and no incident edge at $u$ in $\mathcal{T}$ is of type $\langle$\texttt{tp1}, \texttt{t11}, \texttt{t1q}$\rangle$}{ \label{alg:portoneDFS:tpqCheck}
                \texttt{type}$(u) \gets$  \textbf{partiallyVisited}\; \label{alg:portoneDFS:markPartial}
                $S.pop()$\; \label{alg:portoneDFS:popPartial}
            }{
                $parent(v) \gets u$\; \label{alg:portoneDFS:addEdge}
                $S.push(v)$\; \label{alg:portoneDFS:pushV}
                \texttt{type}$(v) \gets$  \textbf{visited}\; \label{alg:portoneDFS:markVisited}
            }
        }{
            \texttt{type}$(u) \gets $\textbf{fullyVisited}\;
            $S.pop()$\; \label{alg:portoneDFS:popStack}
        }
    }
    \end{algorithm}
\vfill

\pagebreak
\subsection{Pseudocode of Algorithm \ref{alg:vacant} \texttt{Can\_Vacate}()}\label{app:vacant}
\begin{algorithm}[H]
    \caption{\texttt{Can\_Vacate()}}\label{alg:vacant}
    \KwIn{Agent $\psi(x)$ at node $x$}
    \KwOut{State of $\psi(x)$}
    \uIf{$\psi(x).\textsf{parentPort} = \bot$}{
        \Return \emph{settled}\;
    }
    \uElseIf{$\psi(x).\textsf{nodeType} = \textbf{visited}$}{
        visit port 1 neighbor $w$\;
        \If{$\xi(w) \neq \bot$}{
            $\psi(w) \gets \xi(w)$\;
            set $\psi(w).\textsf{vacatedNeighbor} = true$\;
            return to $x$\;
            \Return \emph{settledScout}\;
        }
        \Return \emph{settled}\;
    }
    \uElseIf{$\psi(x).\textsf{nodeType} =$ \textbf{fullyVisited} and $\psi(x).\textsf{vacatedNeighbor} = false$}{
        \Return \emph{settledScout}\;
    }
    \uElseIf{$\psi(x).\textsf{nodeType} = $ \textbf{partiallyVisited}}{
        \Return \emph{settledScout}\;
    }
    \ElseIf{$\psi(x).\textsf{portAtParent} = 1$}{
        Visit parent $z$ of $x$\;
        \eIf{$\psi(z).\textsf{vacatedNeighbor} = false$}{
            $\psi(z).\textsf{state} \gets $\emph{settledScout}\;
            $\psi(z)$ joins $A_{vacated}$\;
            return to $x$\;
            set $\psi(x).\textsf{vacatedNeighbor} = true$\;
            \Return \emph{settled}\;
        }{
            return to $x$\;
            \Return \emph{settled}\;
        }
    }
\end{algorithm}
\vfill

\pagebreak
\subsection{Pseudocode of Algorithm \ref{proc:probe} \texttt{Parallel\_Probe}()}\label{app:probe}
\begin{minipage}[t]{\textwidth}
\resizebox{!}{0.48\textheight}{%
\begin{minipage}{1.3\textwidth}
\begin{algorithm}[H]
  \caption{\texttt{Parallel\_Probe()}}\label{proc:probe}
  \KwIn{Current DFS-head $x$ with settled agent $\psi(x)$, and $A_{scout}$}
  \KwOut{Next port $p_{xy}$}
  $\psi(x).\textsf{probeResult} \gets \bot$;
  $\psi(x).\textsf{checked} \gets 0$\;
  \While{$\psi(x).\textsf{checked} < \delta_x$}{
    $A_{scout} = \{a_1, \dots, a_s\}$ in the increasing order ID\;
    $\Delta' \gets \min(s, \delta_x - \psi(x).\textsf{checked})$\;
    \For{$j \gets 1$ to $\Delta'$}{
    $a \gets$ next agent in $A_{scout}$\;
    \If{ $\psi(x).\textsf{parent}.\textsf{Port} = j + \psi(x).\textsf{checked}$}{
        $j \gets j+1$; $\Delta' \gets \min(s + 1, \delta_x - \psi(x).\textsf{checked})$\;
    }
    $a.\textsf{scoutPort} \gets j + \psi(x).\textsf{checked}$\;
    move via $a.\textsf{scoutPort}$ to reach $y$\;
    $a.\textsf{scoutEdgeType} \gets \textsf{type}(\{x,y\})$\;

    \eIf{$\xi(y)\neq\bot$}{
      $\psi(y)\gets\xi(y)$\;
      $a$ returns to $x$\;
    }{
      \eIf{$\xi(y)=\bot \land p_{yx}=1$}{
        $\psi(y)\gets\bot$\;
        $a$ returns to $x$\;
      }{
        $z \gets$ port-1 neighbor of $y$\;
        $a.\textsf{scoutP1Neighbor} \gets \xi(z)$\;
        $a.\textsf{scoutPortAtP1Neighbor} \gets p_{zy}$\;

        \eIf{$\xi(z)\neq\bot$}{
          $a$ returns to $x$\;
          check $\exists\,b\in A_{scout}: b.\textsf{scoutP1Neighbor}=\xi(z) \land b.\textsf{scoutPortAtP1Neighbor}=p_{zy}$\;
          \eIf{ $b$ found }{
            $\psi(y)\gets b$
          }{
            $\psi(y)\gets\bot$
          }
        }{
          \eIf{$\xi(z)=\bot \land p_{zy}=1$}{
            $\psi(y)\gets\bot$\;
            $a$ returns to $x$\;
          }{
            $w \gets$ port-1 neighbor of $z$\;
            $a.\textsf{scoutP1P1Neighbor} \gets \xi(w)$\;
            $a.\textsf{scoutPortAtP1P1Neighbor} \gets p_{wz}$\;

            \eIf{$\xi(w)=\bot$}{
              $\psi(y)\gets\bot$
            }{
              $a$ returns to $x$\;
              check $\exists\,c\in A_{scout}: c.\textsf{scoutP1Neighbor}=\xi(w) \land c.\textsf{scoutPortAtP1Neighbor}=p_{wz}$\;
              \eIf{$c$ found}{
              check $\exists\,b\in A_{scout}: b.\textsf{scoutP1Neighbor}=c \land  b.\textsf{scoutPortAtP1Neighbor}=p_{zy}$\;
              \eIf{$b$ found}{
                $\psi(y)\gets b$ 
              }{
              $\psi(y)\gets\bot$
              }
              }
              {
              $\psi(y)\gets\bot$
              }
            }
          }
        }
      }
    }
    $a.\textsf{scoutResult} \gets \langle p_{xy},\;a.\textsf{scoutEdgeType},\;\psi(y).\textsf{nodeType},\;
    \psi(y)\rangle$\;
  }
  $\psi(x).\textsf{checked} \gets \psi(x).\textsf{checked} + \Delta'$\;
  $\psi(x).\textsf{probeResult} \gets$ highest priority edge from $a\in A_{scout}$ based on $\psi(y).\textsf{nodeType}$\;
  }
  \Return $p_{xy}$ from $\psi(x).\textsf{probeResult}$\;
\end{algorithm}
\end{minipage}
}
\end{minipage}

\pagebreak
\subsection{Pseudocode of Algorithm \ref{alg:rooted} \texttt{RootedAsync}()}\label{app:rooted}
\begin{minipage}[t]{\textwidth}
\resizebox{!}{0.48\textheight}{%
\begin{algorithm}[H]
    \caption{\texttt{RootedAsync()}}
    \label{alg:rooted}
    \KwIn{A set of $k$ agents at root node $v_0$ in $G$}
    $A \gets$ set of agents\;
    For each $a \in A$, initialize all variables to $\bot$\;
    \For{$a \in A$}{
        $a.\textsf{state} \gets$ \emph{unsettled};
    }
    $A_{unsettled} \gets A$\;
    $A_{vacated} \gets \emptyset$\;
    \While{$A_{unsettled} \neq \emptyset$}{
        $v \gets$ current node\;
        $A_{scout} \gets A_{unsettled} \cup A_{vacated}$\;
        $a_{min} \gets $ Lowest ID agent in $A_{scout}$ \;
        \If{there is no settled agent in $v$}{
            $\psi(v) \gets$ agent with highest ID in $A_{unsettled}$\;
            $\psi(v).\textsf{state} \gets$ \emph{settled}\;
            $\psi(v).\textsf{parent} \gets ( a_{min}.\textsf{prevID},a_{min}.\textsf{childPort})$\;
            $a_{min}.\textsf{childPort} \gets \bot$\;
            $\psi(v).\textsf{parentPort} \gets a_{min}.\textsf{arrivalPort}$\;
            $A_{unsettled} \gets A_{unsettled} - \{\psi(v)\}$\;
            \If{$A_{unsettled} = \emptyset$}{
                break\;
            }
        }
        $a_{min}.\textsf{prevID} \gets \psi(v).\textsf{ID}$\;
        \If{$\delta_v \ge k-1$}{
            run \texttt{Parallel\_Probe}$(\psi(v), A_{scout})$ for $k-1$ ports\;
            send unsettled agents to empty neighbors\;
            break\;
        }
        $\psi(v).\textsf{sibling} \gets a_{min}.\textsf{siblingDetails}$\;
        $a_{min}.\textsf{siblingDetails} \gets \bot$\;
        nextPort $\gets$ \texttt{Parallel\_Probe}$(\psi(x), A_{scout})$\;
        $\psi(v).\textsf{state} \gets$ \texttt{Can\_Vacate()}\;
        \If{$\psi(v).\textsf{state} =${settledScout}}{
            $A_{vacated} \gets A_{vacated}\cup \{\psi(v)\}$\;
            $A_{scout} \gets A_{unsettled} \cup A_{vacated}$\;
        }
        \eIf{nextPort $\neq \bot$}{
            $\psi(v).\textsf{recentPort} \gets$ nextPort\;
            $a_{min}.\textsf{childPort} \gets$ nextPort\;
            \eIf{$\psi(v).\textsf{recentChild} = \bot$}{
                $\psi(v).\textsf{recentChild} \gets$ nextPort\;
                }{
                    $a_{min}.\textsf{siblingDetails} \gets a_{min}.\textsf{childDetails}$\;
                    $a_{min}.\textsf{childDetails} \gets \bot$\;
                    $\psi(v).\textsf{recentChild} \gets$ nextPort\;
                }
            All agents in $A_{scout}$ move through nextPort\;
        }
        {
            $a_{min}.\textsf{childDetails} \gets (\psi(v).\textsf{ID}, \psi(v).\textsf{portAtParent})$\;
            $a_{min}.\textsf{childPort} \gets \bot$\;
            $\psi(v).\textsf{recentPort} \gets \psi(v).\textsf{parentPort}$\;
            All agents in $A_{scout}$ move though $\psi(v).\textsf{parentPort}$\;
        }
    }
    \texttt{Retrace}($A_{vacated}$)\;
    \end{algorithm}
}
\end{minipage}
\vfill
\pagebreak
\subsection{Pseudocode of Algorithm \ref{alg:retrace} \texttt{Retrace}()}\label{app:retrace}
\begin{minipage}[t]{\textwidth}
\resizebox{\textwidth}{!}{%
\begin{algorithm}[H]
    \caption{\texttt{Retrace()}}
    \label{alg:retrace}
    \KwIn{$A_{vacated}$ - set of agents with state \emph{settledScout}}
    
    \While{$A_{vacated} \neq \emptyset$}{
        $v \gets$ current node\;
        $a_{min} \gets $ Lowest ID agent in $A_{vacated}$ \;
        \If{$\xi(v) = \bot$}{ \tcp{no settled agent present at $v$, it must be in $A_{vacated}$}
            find $a \in A_{vacated}$ with $a.\textsf{ID} = a_{min}.\textsf{nextAgentID}$ at $v$\;
                $a.\textsf{state} \gets$ \emph{settled}\;
                $A_{vacated} \gets A_{vacated} - \{a\}$\;
                $a_{min} \gets$ Lowest ID agent in $A_{vacated}$\;
                $\psi(v) \gets a$\;
        }

        \tcp{If all agents are settled, retrace is complete}
        \If{$A_{vacated} = \emptyset$}{
            break\;
        }
        
        \tcp{Determine next move in post-order traversal}
        \eIf{$\psi(v).\textsf{recentChild} \neq \bot$}{
            \eIf{$\psi(v).\textsf{recentChild} = a_{min}.\textsf{arrivalPort}$}{ 
                \eIf{$a_{min}.\textsf{siblingDetails} = \bot$}{
                    $\psi(v).\textsf{recentChild} \gets \bot$\;
                    $(a_{min}.\textsf{nextAgentID}, a_{min}.\textsf{nextPort}) \gets \psi(v).parent$\;
                    $a_{min}.\textsf{siblingDetails} \gets \psi(v).\textsf{sibling}$\;
                }{
                    $(a_{min}.\textsf{nextAgentID}, a_{min}.\textsf{nextPort}) \gets a_{min}.\textsf{siblingDetails}$\;
                    $a_{min}.\textsf{siblingDetails} \gets \bot$\;
                    $\psi(v).\textsf{recentChild} \gets a_{min}.\textsf{nextPort}$\;
                }
            }
            {
                $a_{min}.\textsf{nextPort} \gets \psi(v).\textsf{recentChild}$\;
                Check if $\exists a \in A_{vacated}: a.\textsf{parent} = (\psi(v).\textsf{ID}, \psi(v).\textsf{recentChild})$\;
                \If{$a$ found}{
                    $a_{min}.\textsf{nextAgentID} \gets a.\textsf{ID}$\;
                    $a_{min}.\textsf{nextPort} \gets \psi(v).\textsf{recentChild}$\;
                }
            }
        }{
            $(parentID, portAtParent) \gets \psi(v).\textsf{parent}$\;
            $a_{min}.\textsf{nextAgentID} \gets parentID$\;
            $a_{min}.\textsf{nextPort} \gets \psi(v).\textsf{parentPort}$\;
            $a_{min}.\textsf{siblingDetails} \gets \psi(v).\textsf{sibling}$\;
        }
        All agents in $A_{vacated}$ move through $a_{min}.\textsf{nextPort}$\;
    }
    \end{algorithm}
}
\end{minipage}
\vfill
\end{toappendix}

%% file: arxiv.bbl
\begin{thebibliography}{10}

\bibitem{Augustine:2018}
John Augustine and William~K. {Moses Jr.}
\newblock Dispersion of mobile robots: {A} study of memory-time trade-offs.
\newblock In {\em ICDCN}, pages 1:1--1:10, 2018.

\bibitem{Bampas:2009}
Evangelos Bampas, Leszek G{a}sieniec, Nicolas Hanusse, David Ilcinkas, Ralf Klasing, and Adrian Kosowski.
\newblock Euler tour lock-in problem in the rotor-router model: I choose pointers and you choose port numbers.
\newblock In {\em DISC}, pages 423--435, 2009.

\bibitem{BKM24}
Rik Banerjee, Manish Kumar, and Anisur~Rahaman Molla.
\newblock Optimizing robot dispersion on unoriented grids: With and without fault tolerance.
\newblock In Quentin Bramas, Arnaud Casteigts, and Kitty Meeks, editors, {\em ALGOWIN}, pages 31--45. Springer, 2024.

\bibitem{BKM25}
Rik Banerjee, Manish Kumar, and Anisur~Rahaman Molla.
\newblock Optimal fault-tolerant dispersion on oriented grids.
\newblock In Amos Korman, Sandip Chakraborty, Sathya Peri, Chiara Boldrini, and Peter Robinson, editors, {\em ICDCN}, pages 254--258. {ACM}, 2025.

\bibitem{Barriere2009}
L.~Barriere, P.~Flocchini, E.~Mesa-Barrameda, and N.~Santoro.
\newblock Uniform scattering of autonomous mobile robots in a grid.
\newblock In {\em IPDPS}, pages 1--8, 2009.

\bibitem{brawley1977territorial}
Susan~H Brawley and Walter~H Adey.
\newblock Territorial behavior of threespot damselfish (eupomacentrus planifrons) increases reef algal biomass and productivity.
\newblock {\em Environmental Biology of Fishes}, 2:45--51, 1977.

\bibitem{ChandKMS23}
Prabhat~Kumar Chand, Manish Kumar, Anisur~Rahaman Molla, and Sumathi Sivasubramaniam.
\newblock Fault-tolerant dispersion of mobile robots.
\newblock In Amitabha Bagchi and Rahul Muthu, editors, {\em CALDAM}, pages 28--40, 2023.

\bibitem{Cohen:2008}
Reuven Cohen, Pierre Fraigniaud, David Ilcinkas, Amos Korman, and David Peleg.
\newblock Label-guided graph exploration by a finite automaton.
\newblock {\em ACM Trans. Algorithms}, 4(4):42:1--42:18, August 2008.

\bibitem{Cord-LandwehrDFHKKKKMHRSWWW11}
Andreas Cord{-}Landwehr, Bastian Degener, Matthias Fischer, Martina H{\"{u}}llmann, Barbara Kempkes, Alexander Klaas, Peter Kling, Sven Kurras, Marcus M{\"{a}}rtens, Friedhelm {Meyer auf der Heide}, Christoph Raupach, Kamil Swierkot, Daniel Warner, Christoph Weddemann, and Daniel Wonisch.
\newblock A new approach for analyzing convergence algorithms for mobile robots.
\newblock In {\em ICALP}, pages 650--661, 2011.

\bibitem{Cybenko:1989}
G.~Cybenko.
\newblock Dynamic load balancing for distributed memory multiprocessors.
\newblock {\em J. Parallel Distrib. Comput.}, 7(2):279--301, October 1989.

\bibitem{DasCALDAM21}
Archak Das, Kaustav Bose, and Buddhadeb Sau.
\newblock Memory optimal dispersion by anonymous mobile robots.
\newblock In Apurva Mudgal and C.~R. Subramanian, editors, {\em CALDAM}, volume 12601, pages 426--439. Springer, 2021.

\bibitem{Das18}
Shantanu Das, Dariusz Dereniowski, and Christina Karousatou.
\newblock Collaborative exploration of trees by energy-constrained mobile robots.
\newblock {\em Theory Comput. Syst.}, 62(5):1223--1240, 2018.

\bibitem{Das16}
Shantanu Das, Paola Flocchini, Giuseppe Prencipe, Nicola Santoro, and Masafumi Yamashita.
\newblock Autonomous mobile robots with lights.
\newblock {\em Theor. Comput. Sci.}, 609:171--184, 2016.

\bibitem{Dereniowski:2015}
Dariusz Dereniowski, Yann Disser, Adrian Kosowski, Dominik Pajak, and Przemyslaw Uzna\'{n}ski.
\newblock Fast collaborative graph exploration.
\newblock {\em Inf. Comput.}, 243(C):37--49, August 2015.

\bibitem{ElorB11}
Yotam Elor and Alfred~M. Bruckstein.
\newblock Uniform multi-agent deployment on a ring.
\newblock {\em Theor. Comput. Sci.}, 412(8-10):783--795, 2011.

\bibitem{Fraigniaud:2005}
Pierre Fraigniaud, David Ilcinkas, Guy Peer, Andrzej Pelc, and David Peleg.
\newblock Graph exploration by a finite automaton.
\newblock {\em Theor. Comput. Sci.}, 345(2-3):331--344, November 2005.

\bibitem{GorainSSS22}
Barun Gorain, Partha~Sarathi Mandal, Kaushik Mondal, and Supantha Pandit.
\newblock Collaborative dispersion by silent robots.
\newblock In {\em SSS}, volume 13751, pages 254--269. Springer, 2022.

\bibitem{ItalianoPS22}
Giuseppe~F. Italiano, Debasish Pattanayak, and Gokarna Sharma.
\newblock Dispersion of mobile robots on directed anonymous graphs.
\newblock In Merav Parter, editor, {\em SIROCCO}, pages 191--211. Springer, 2022.

\bibitem{KaurD2D23}
Tanvir Kaur and Kaushik Mondal.
\newblock Distance-2-dispersion: {{Dispersion}} with further constraints.
\newblock In David Mohaisen and Thomas Wies, editors, {\em Networked Systems}, pages 157--173, 2023.

\bibitem{KshemkalyaniICDCN19}
Ajay~D. Kshemkalyani and Faizan Ali.
\newblock Efficient dispersion of mobile robots on graphs.
\newblock In {\em ICDCN}, pages 218--227, 2019.

\bibitem{kshemkalyani2025dispersion}
Ajay~D Kshemkalyani, Manish Kumar, Anisur~Rahaman Molla, Debasish Pattanayak, and Gokarna Sharma.
\newblock Dispersion is (almost) optimal under (a) synchrony.
\newblock In {\em SPAA}, 2025.

\bibitem{KshemkalyaniKMS24}
Ajay~D. Kshemkalyani, Manish Kumar, Anisur~Rahaman Molla, and Gokarna Sharma.
\newblock Brief announcement: Agent-based leader election, mst, and beyond.
\newblock In Dan Alistarh, editor, {\em DISC}, volume 319 of {\em LIPIcs}, pages 50:1--50:7, 2024.

\bibitem{KshemkalyaniALGOSENSORS19}
Ajay~D. Kshemkalyani, Anisur~Rahaman Molla, and Gokarna Sharma.
\newblock Fast dispersion of mobile robots on arbitrary graphs.
\newblock In {\em ALGOSENSORS}, pages 23--40, 2019.

\bibitem{KshemkalyaniWALCOM20}
Ajay~D. Kshemkalyani, Anisur~Rahaman Molla, and Gokarna Sharma.
\newblock Dispersion of mobile robots on grids.
\newblock In {\em {WALCOM}}, pages 183--197, 2020.

\bibitem{KshemkalyaniICDCS20}
Ajay~D. Kshemkalyani, Anisur~Rahaman Molla, and Gokarna Sharma.
\newblock Efficient dispersion of mobile robots on dynamic graphs.
\newblock In {\em ICDCS}, pages 732--742, 2020.

\bibitem{KshemkalyaniICDCN20}
Ajay~D. Kshemkalyani, Anisur~Rahaman Molla, and Gokarna Sharma.
\newblock Dispersion of mobile robots using global communication.
\newblock {\em J. Parallel Distributed Comput.}, 161:100--117, 2022.

\bibitem{KshemkalyaniOPODIS21}
Ajay~D. Kshemkalyani and Gokarna Sharma.
\newblock Near-optimal dispersion on arbitrary anonymous graphs.
\newblock In Quentin Bramas, Vincent Gramoli, and Alessia Milani, editors, {\em OPODIS}, volume 217 of {\em LIPIcs}, pages 8:1--8:19, 2021.

\bibitem{le1999neural}
Nicole Le~Douarin and Chaya Kalcheim.
\newblock {\em The neural crest}.
\newblock Number~36. Cambridge university press, 1999.

\bibitem{MencPU17}
Artur Menc, Dominik Pajak, and Przemyslaw Uznanski.
\newblock Time and space optimality of rotor-router graph exploration.
\newblock {\em Inf. Process. Lett.}, 127:17--20, 2017.

\bibitem{tamc19}
Anisur~Rahaman Molla and William K.~Moses Jr.
\newblock Dispersion of mobile robots: The power of randomness.
\newblock In {\em TAMC}, pages 481--500, 2019.

\bibitem{MollaIPDPS21}
Anisur~Rahaman Molla, Kaushik Mondal, and William K.~Moses Jr.
\newblock Byzantine dispersion on graphs.
\newblock In {\em IPDPS}, pages 942--951. {IEEE}, 2021.

\bibitem{MollaMM21}
Anisur~Rahaman Molla, Kaushik Mondal, and William K.~Moses Jr.
\newblock Optimal dispersion on an anonymous ring in the presence of weak byzantine robots.
\newblock {\em Theor. Comput. Sci.}, 887:111--121, 2021.

\bibitem{Pattanayak-WDALFR20}
Debasish Pattanayak, Gokarna Sharma, and Partha~Sarathi Mandal.
\newblock Dispersion of mobile robots tolerating faults.
\newblock In {\em WDALFR}, pages 17:1--17:6, 2021.

\bibitem{Poudel18}
Pavan Poudel and Gokarna Sharma.
\newblock Time-optimal uniform scattering in a grid.
\newblock In {\em ICDCN}, pages 228--237, 2019.

\bibitem{SaxenaK025}
Ashish Saxena, Tanvir Kaur, and Kaushik Mondal.
\newblock Dispersion on time varying graphs.
\newblock In Amos Korman, Sandip Chakraborty, Sathya Peri, Chiara Boldrini, and Peter Robinson, editors, {\em ICDCN}, pages 269--273. {ACM}, 2025.

\bibitem{Saxena025}
Ashish Saxena and Kaushik Mondal.
\newblock Path connected dynamic graphs with a study of efficient dispersion.
\newblock In Amos Korman, Sandip Chakraborty, Sathya Peri, Chiara Boldrini, and Peter Robinson, editors, {\em ICDCN}, pages 171--180. {ACM}, 2025.

\bibitem{Shibata:2016}
Masahiro Shibata, Toshiya Mega, Fukuhito Ooshita, Hirotsugu Kakugawa, and Toshimitsu Masuzawa.
\newblock Uniform deployment of mobile agents in asynchronous rings.
\newblock In {\em PODC}, pages 415--424, 2016.

\bibitem{ShintakuSKM20}
Takahiro Shintaku, Yuichi Sudo, Hirotsugu Kakugawa, and Toshimitsu Masuzawa.
\newblock Efficient dispersion of mobile agents without global knowledge.
\newblock In {\em SSS}, volume 12514, pages 280--294, 2020.

\bibitem{Subramanian:1994}
Raghu Subramanian and Isaac~D. Scherson.
\newblock An analysis of diffusive load-balancing.
\newblock In {\em SPAA}, pages 220--225, 1994.

\bibitem{sudo24}
Yuichi Sudo, Masahiro Shibata, Junya Nakamura, Yonghwan Kim, and Toshimitsu Masuzawa.
\newblock Near-linear time dispersion of mobile agents.
\newblock In Dan Alistarh, editor, {\em DISC}, pages 38:1--38:22.

\bibitem{wurman2008coordinating}
Peter~R Wurman, Raffaello D'Andrea, and Mick Mountz.
\newblock Coordinating hundreds of cooperative, autonomous vehicles in warehouses.
\newblock {\em AI magazine}, 29(1):9--9, 2008.

\end{thebibliography}
